\documentclass[a4paper,fleqn]{cas-dc}

\usepackage[numbers]{natbib}
\usepackage[boxed,commentsnumbered,ruled,linesnumbered]{algorithm2e}
\usepackage{enumitem}
\usepackage{hyperref}
\newtheorem{definition}{Definition}  
\newtheorem{proof}{Proof}    
\newtheorem{theorem}{Theorem}   

\setlist[itemize]{nosep}
\hypersetup{
	colorlinks=true,  
	linkcolor=red,   
	citecolor=blue,    
	urlcolor=black,   
}

\begin{document}
	
\shorttitle{Cross-level Privacy Preserving Utility Mining} 
\shortauthors{Jiahong Cai et~al.}	
\title [mode = title]{Cross-level Privacy Preserving Utility Mining}                      	
\author[1]{Jiahong Cai}
\ead{jhcai321@gmail.com}
\address[1]{College of Cyber Security, Jinan University, Guangzhou 510632, P.R. China}
	
\author[1]{Wensheng Gan}
\cormark[1]
\ead{wsgan001@gmail.com}
\cortext[cor1]{Corresponding author}
	
\author[2]{Philip S. Yu}
\ead{psyu@uic.edu}
\address[2]{Department of Computer Science, University of Illinois Chicago, Chicago IL 60607, USA.}
	
\begin{abstract}
    Privacy-preserving utility mining (PPUM) aims to hide sensitive high-utility patterns while preserving the utility of the sanitized database. In practice, however, many datasets are associated with taxonomic information, which makes the identification and processing of generalized items more challenging. To address this, we investigate the cross-level privacy-preserving utility mining (CLPPUM) problem and propose a method for protecting generalized items. Based on different victim item selection strategies, we develop three CLPPUM algorithms: minimum \textit{RGISU} first (Min-RF), maximum \textit{RGISU} first (Max-RF), and best \textit{NSC} first (Best-NSCF). Furthermore, to enable efficient victim item identification, a novel dictionary structure named \textit{GI-dic} is designed to accelerate the computation of required utility metrics. Experimental results on multiple datasets demonstrate that the proposed algorithms successfully hide all sensitive cross-level high-utility itemsets without introducing artificial itemsets. The results also show that our method performs well on sparse datasets, and both Min-RF and Best-NSCF consistently outperform Max-RF. Overall, Min-RF achieves the best performance, particularly when the minimum utility threshold is low and the dataset is dense. Datasets and code are available at \url{https://github.com/jhcai321/CLPPUM}.
\end{abstract}

\begin{keywords}
    privacy-preserving \sep cross-level itemsets \sep sensitive pattern \sep utility mining \sep taxonomy
\end{keywords}
	
\maketitle
	
\section{Introduction} \label{sec:introduction}
	
With the rapid growth of large-scale databases, the volume of data has increased significantly. Data mining \cite{gan2017data,gan2020huopm} has therefore become an important tool for extracting useful information from large-scale datasets. For example, enterprises analyze transaction records and user behavior to improve decision-making and enhance profitability. In this context, data mining has been expanded into different research directions. Frequent itemset mining (FIM) \cite{agrawal1993mining} focuses on the combination of goods that frequently appear in transaction records. However, in real life, low-frequency combinations of goods are still very likely to generate large profits. To overcome the limitations of frequency pattern mining, the concept of high-utility itemset mining (HUIM) \cite{gan2021survey,yao2004foundational} was introduced. In HUIM, the utility of an item is determined by internal utility (quantity) and external utility (unit profit/importance). Therefore, it can be better used to analyze factors such as user preferences, importance, and profit in reality. A variety of HUIM algorithms have been developed \cite{liu2012mining,liu2005two,zida2017efim}. Although HUIM has been widely used, it ignores the concept of hierarchical structure in the real world. As shown in Fig. \ref{peripherals_taxonomy}, in computer peripherals, “keyboard” and “mouse” are the materialization of the abstract concept of “input device”. Therefore, to enable data mining to consider hierarchical information, the cross-level high-utility itemset mining (CLHUIM) \cite{cagliero2017discovering,fournier2020mining,tung2022efficient} has been proposed.

\begin{figure}[ht]
	\centering
	\includegraphics[clip,width=0.48\textwidth]{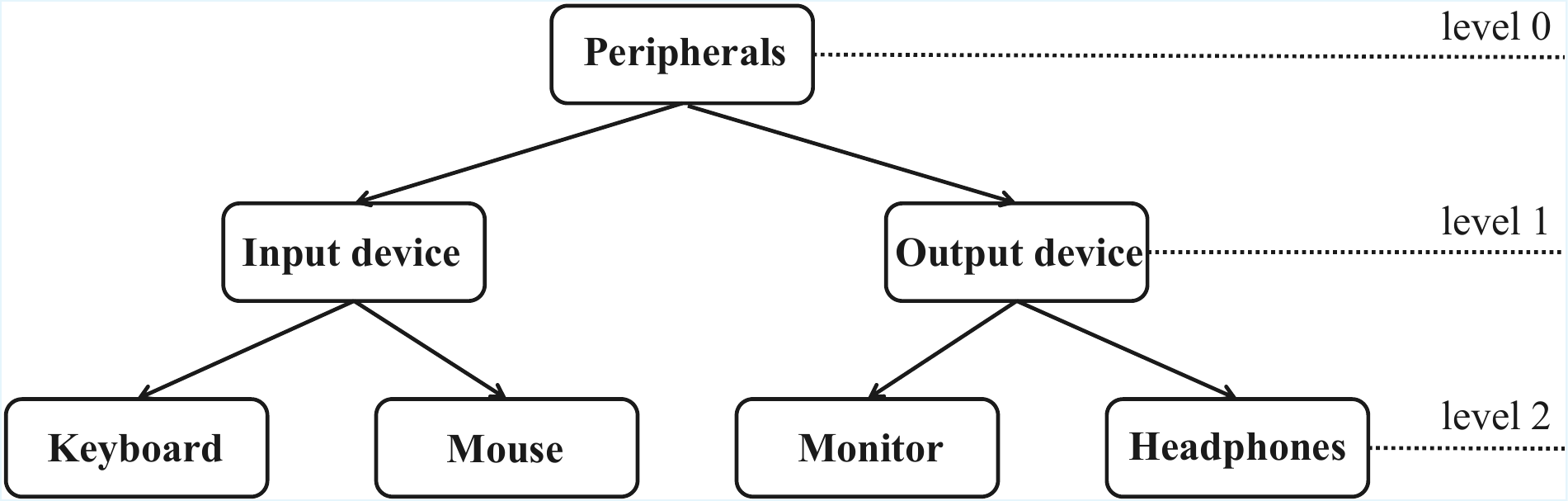} 
	\caption{A taxonomy tree of computer peripherals.}
	\label{peripherals_taxonomy} 
\end{figure}
	
HUIM can assist decision-makers in analyzing massive real-world datasets and developing high-profit strategies from discovered itemsets. However, in data analysis in medical, government, and commercial organizations, data is often transferred between data providers and data miners. Analyzing such data may lead to the leakage of sensitive information held by data providers. For example, in marketing scenarios, if the high-value discovered information reveals the consumption habits or behaviors of a particular gender, race, or group, then this information is likely to be sensitive, and its leakage may raise ethical concerns. Therefore, privacy-preserving data mining  (PPDM) \cite{agrawal2000privacy,lindell2000privacy} was proposed. PPDM minimizes the impact on non-sensitive information while hiding sensitive information, thereby preserving data availability. Privacy-preserving utility mining (PPUM) \cite{gan2018privacy} is a branch of PPDM and can be divided into heuristic-based \cite{lin2017efficient,lin2016fast,yeh2010hhuif}, border-based \cite{liu2018novel}, and exact \cite{li2019novel,nguyen2023new} algorithms. Existing PPUM methods are effective for hiding sensitive high-utility patterns in conventional settings, but they do not explicitly consider hierarchical information. In real-world data, itemsets may involve generalized items, which increases the complexity of representation and processing. As a result, directly applying existing PPUM methods to the cross-level setting may not effectively hide sensitive cross-level high-utility itemsets. This limitation motivates the study of cross-level privacy-preserving utility mining (CLPPUM), which aims to handle sensitive itemsets that contain items from different levels. The specific contributions of this paper are as follows:

\begin{itemize}
    \item This paper defines a new task, cross-level privacy-preserving utility mining (CLPPUM), for identifying and hiding cross-level high-utility itemsets (CLHUIs).
		
    \item Three indicators—sensitive count (SC), non-sensitive count (NSC), and real item sensitive utility (RISU)—are used to measure the relationship between items and sensitive/non-sensitive CLHUIs. These indicators are extended in the cross-level context, and a dictionary structure \textit{GI-dic} is designed to accelerate the calculation of indicators and assist in the selection of victim items.
		
    \item Three algorithms are designed to hide sensitive CLHUIs. The algorithms are named Minimum \textit{RGISU} First (Min-RF), Maximum \textit{RGISU} First (Max-RF), and Best \textit{NSC} First (Best-NSCF). The study also compares the effects of different indicators on SCLHUIs.
		
    \item Experiments are performed on several datasets to evaluate the performance of the proposed algorithms. The results indicate that the proposed algorithms perform well on sparse datasets, correctly hiding all sensitive CLHUIs without introducing any false CLHUIs into the database. As a result, both the hiding failure (HF) and the artificial cost (AC) are 0. Among the three proposed algorithms, Min-RF exhibits the optimal performance.		
\end{itemize}
	
The rest of this paper is arranged as follows: Section \ref{sec:relatedwork} presents and summarizes the research work related to PPUM. Section \ref{sec:preliminaries} introduces the basic knowledge of CLHUIM and PPUM. Section \ref{sec:algorithm} provides a detailed exposition of the three CLPPUM algorithms proposed in this paper and analyzes their privacy-related properties. Section \ref{sec:experiment} conducts experiments and analyzes the results. Section \ref{sec:conclusion} summarizes this paper and looks forward to future work.
	
\section{Related Work} \label{sec:relatedwork}
\subsection{High-utility itemset mining}
		
To consider both the number of items and their unit profit/weight in the data mining process, frequent itemset mining (FIM) has been extended to high-utility itemset mining (HUIM) \cite{gan2021survey,yao2004foundational}. However, since utility does not have anti-monotonicity with frequency, the HUIM task is more challenging \cite{liu2005two}. To address this issue, many studies introduced the concept of upper bounds and proposed many efficient HUIM algorithms to improve mining efficiency. Early work proposed the Two-Phase algorithm \cite{liu2005two}, which relies on the transaction-weighted utility (TWU) upper bound. Based on the TWU, several tree-based two-phase approaches were developed, including IHUP \cite{ahmed2009efficient}, UP-Growth \cite{tseng2010up}, and UP-Growth+ \cite{tseng2012efficient}. Although these methods reduce the number of database scans, they still rely on the loose TWU upper bound and tend to generate a large number of candidate itemsets. To overcome this limitation, Liu et al. \cite{liu2012mining} proposed HUI-Miner, a one-phase algorithm based on the utility-list. By adopting a tighter remaining utility upper bound, it avoids candidate generation altogether. Subsequent methods, such as D2HUP \cite{liu2015mining}, FHM \cite{fournier2014fhm}, and EFIM \cite{zida2017efim}, further improved efficiency through tighter upper bounds, more effective pruning strategies, and techniques such as high-utility database projection (HDP) and high-utility transaction merging (HTM). More recent studies have also explored memory reuse \cite{duong2018efficient}, bitwise acceleration \cite{wu2022ubp}, structural simplification \cite{cheng2023efficient}, and approximate search \cite{qu2023mining, yan2024efficient} to reduce both memory consumption and computational cost during utility-list construction and processing.
	
Traditional HUIM cannot identify hierarchical information in reality. To address this limitation, taxonomy information has been incorporated into HUIM to support hierarchical structures. ML-HUI-Miner \cite{cagliero2017discovering} introduced the concept of generalized HUIs. However, this algorithm can only mine items belonging to the same taxonomy level and cannot identify itemsets that span different levels. Subsequently, CLH-Miner \cite{fournier2020mining} and FEACP \cite{tung2022efficient} were proposed to support cross-level high-utility itemset mining (CLHUIM), using tighter upper bounds to reduce memory usage and runtime. Additionally, variants of CLHUIM algorithms have been proposed to discover top-k CLHUIs \cite{han2024mining,nouioua2020tkc, truong2023efficient} and CLHUIs in databases with unstable and negative profits \cite{tung2025mining}.
	
\subsection{Privacy-preserving utility mining}
	
Utility pattern mining can assist in analyzing high-value information within data, bringing significant benefits to companies. However, when enterprises store, use, or share business and healthcare data, there is a high risk of disclosing sensitive or confidential information, which may lead to privacy concerns \cite{wu2021hiding}. As a result, privacy-preserving data mining (PPDM) has been developed into privacy-preserving utility mining (PPUM). These algorithms can hide sensitive information in a database while balancing privacy protection and data sharing. PPUM algorithms can be divided into heuristic-based, border-based, and exact algorithms. Yeh and Hsu \cite{yeh2010hhuif} first formalized the PPUM problem and proposed two classical heuristic-based algorithms, HHUIF and MSICF. To improve efficiency, a tree-based method called FPUTT \cite{yun2015fast} organizes the database to accelerate the identification and hiding of sensitive itemsets. Yin et al. \cite{yin2023fast} proposed FULD, which utilizes a novel utility dictionary structure to enhance performance. Subsequent studies have explored more sophisticated strategies for selecting victim items. Lin et al. \cite{lin2016fast} proposed the MSU-MAU and MSU-MIU algorithms, incorporating multiple similarity measures to evaluate the side effects of data sanitization. Lin et al. \cite{liu2020improved} further considered non-sensitive itemsets in the victim selection process and introduced the IMSICF algorithm, effectively reducing the impact on non-sensitive patterns. The MinMax and Weighted algorithms \cite{jangra2022efficient} combine dual sorting strategies with different victim selection criteria to improve efficiency and reduce side effects. Ashraf et al. \cite{ashraf2023efficient} introduced a new metric, real item sensitive utility (RISU), along with a refined sorting strategy to further optimize performance.

In addition, boundary-based and exact approaches have been proposed to determine optimal sanitization strategies, such as using maximum boundary values \cite{liu2018novel} or formulating the problem as integer programming \cite{li2019novel,nguyen2024novel,nguyen2023new}. However, these methods often suffer from high computational complexity and long execution time \cite{jangra2022efficient}. Although heuristic-based algorithms cannot guarantee the best results, they are easy to understand and implement. In addition, heuristic-based privacy-preserving algorithms have also been extended to other related mining tasks, including frequent high average-utility itemset mining \cite{le2022h}, periodic high-utility pattern mining \cite{zhou2025utility}, rare itemset mining \cite{chen2025rare,gui2024privacy}, and association rule mining \cite{aljehani2025preserving}. Therefore, heuristic-based algorithms remain the focus of most current research. Although numerous algorithms have been developed for hiding HUIs, none of the existing privacy-preserving methods have taken real-world categorical hierarchical information into account. Therefore, this paper defines cross-level PPUM that aims to achieve the hiding of sensitive cross-level high-utility itemsets.

\section{Preliminaries} \label{sec:preliminaries}
	
Let $I$ = \{$i_1$, $i_2$, ..., $i_m$\} be a set of items. A quantitative transaction database is a collection of transactions $\mathcal{D}$ = \{$T_1$, $T_2$, ..., $T_n$\}. For each transaction $T_j$ $\in$ $\mathcal{D}$, $T_j \subseteq I$, each $T_j$ has a unique identifier $j$. For each item $v$ $\in$ $I$, the associated external utility (i.e., unit profit) is denoted as  $p(v)$. For each item $v$ $\in$ $T_j$, its internal utility, referring to the quantity in transaction $T_j$, is represented by $q(v,T_j)$. For example, Table \ref{Tab:transaction} contains six items ($I$ = \{$a$, $b$, $c$, $d$, $e$, $f$\}) and eight transactions ($\mathcal{D}$ = \{$T_1$, $T_2$, ..., $T_8$\}). In the transaction $T_3$, the internal utility of items  $a$, $c$, and $d$ is 1, 5, and 1, respectively. We set the external utility values of \{$a$, $b$, $c$, $d$, $e$, $f$\} to \{5, 1, 3, 3, 2, 1\}. 
	
	\begin{table}
		\centering	\makeatletter\def\@captype{table}\makeatother\setlength{\belowcaptionskip}{10pt}
		\caption{An example quantitative transaction database.}
		\label{Tab:transaction}
		{\normalsize\fontfamily{ptm}\selectfont  
			\begin{tabular}{p{0.5cm} p{4.5cm} p{0.8cm}} 
				\hline
				\textbf{$T_{id}$} & \textbf{Transaction} & \textbf{Utility} \\
				\hline
				$T_1$ & $(a,1)$, $(b,1)$, $(d,1)$ & 9 \\
				$T_2$ & $(a,2)$, $(d,3)$, $(e,1)$ & 21 \\
				$T_3$ & $(a,1)$, $(b,2)$, $(c,5)$, $(d,1)$, $(e,3)$ & 26 \\
				$T_4$ & $(d,4)$, $(e,3)$ & 18 \\
				$T_5$ & $(a,1)$, $(b,1)$, $(d,1)$ & 9 \\
				$T_6$ & $(d,5)$, $(e,2)$, $(f,2)$ & 21 \\
				$T_7$ & $(a,2)$, $(c,1)$ & 13 \\
				$T_8$ & $(a,1)$, $(b,4)$, $(e,3)$ & 15 \\
				\hline
			\end{tabular}
		}
	\end{table}

\begin{definition}
    \label{def:Taxonomy}
	\rm \textbf{(Taxonomy) \cite{cagliero2017discovering,fournier2020mining}}. A taxonomy $\tau$ is a tree structure defined on the quantitative transaction database $\mathcal{D}$. In this structure, each leaf node corresponds to an item $v$ $\in$ $I$. Each internal node denotes a category formed by aggregating all of its descendant leaf nodes, referred to as a generalized item (generalized items). The set of all generalized items is denoted as $GI$, and the union of all generalized items and leaf items is denoted as $AI$, i.e., $AI$ = $GI \cup I$. Let the relationship $LR \subseteq GI \times I$, such that if there exists a path from $g$ to $v$, then $(g, v) \in LR$. Similarly, let the relationship $GR \subseteq AI \times AI$, such that if there exists a path from item $d$ to item $f$, then $(d, f) \in GR$. 
\end{definition}
	
\begin{definition}\label{def:Descendant}
	\rm \textbf{(Descendant) \cite{fournier2020mining}}. In taxonomy $\tau$, the leaf items of a generalized item $g$ are all leaf nodes of $g$ in the tree structure that are reachable from $g$ along a path. This collection is formally defined as $\textit{Leaf}(g, \tau)$ = $\{v \mid (g, v) \in LR\}$. The descendant nodes of a (generalized) item $d$ are referred to as all nodes in the tree structure that are descendants of $d$, and are defined as: $\textit{Desc}(d, \tau)$ = $\{f \mid (d, f) \in GR\}$. $\textit{level}(d)$ denotes the number of edges on the path from the root node to item $d$. A collection of items $P$ is considered an itemset, where $P \subseteq AI \land \nexists$ $i, j$ $\in P, i \in \textit{Desc}(j, \tau)$. If $\exists$ $g \in P$ where $g \in GI$, then $P$ is a generalized itemset. 
\end{definition} 
	
\begin{figure}[!htbp]
	\centering
	\includegraphics[width=0.65\linewidth]{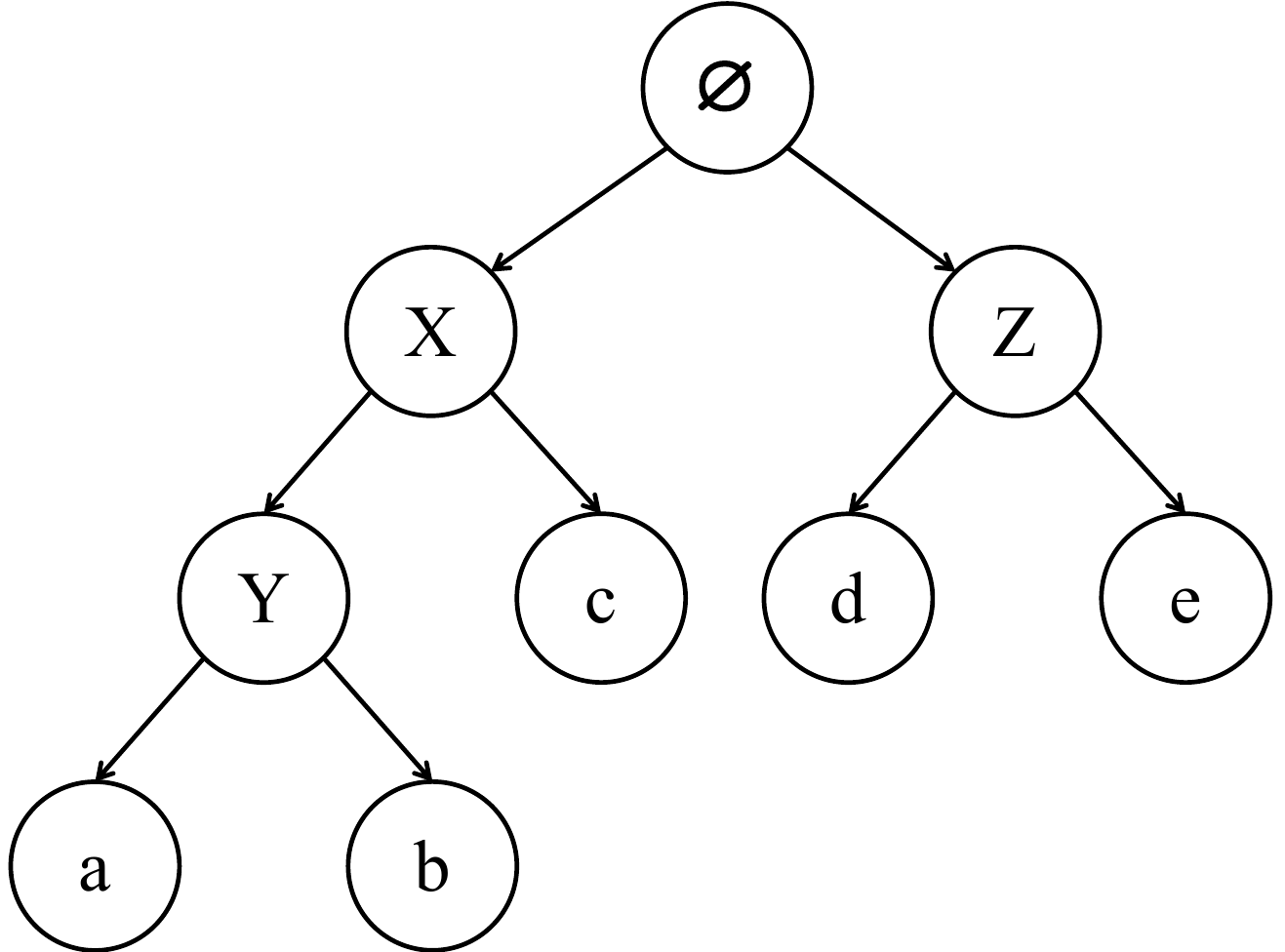} 
	\caption{A taxonomy of items.}
	\label{taxonomy} 
\end{figure}

\begin{definition}\label{def:Utility_of_a_generalized_item/itemset}
	\rm \textbf{(Utility of a generalized item/itemset) \cite{fournier2020mining}}. The utility of an item $v$ in a transaction $T_j$ is calculated as: $u(v, T_j)$ = $q(v, T_j)$ $\times$ $p(v)$. The utility of an itemset $P$ in $T_j$ is calculated as $u(P, T_j)$ =  $\sum_{ v \in P }u(v, T_j)$. The utility of an itemset $P$ in $\mathcal{D}$ is $u(P)$ = $\sum_{ T_j \in g(P) }u(P, T_j)$, where $g(P)$ is the set of transactions in $\mathcal{D}$ that contain $P$. The utility of a generalized item $g$ in a transaction $T_j$ is denoted as $u(g, T_j)$, is calculated as the sum of the utilities of all its leaf items: $\sum_{ v \in \textit{Leaf}(g,\tau) }p(v) \times q(v, T_j)$. The utility of a generalized itemset $\textit{GP}$ in $T_j$ is determined as $u(\textit{GP}, T_j) $ = $\sum_{ d \in \textit{GP} }u(d, T_j)$. The utility of a generalized itemset $\textit{GP}$ in $\mathcal{D}$ is calculated as $u(\textit{GP})$ = $\sum_{ T_j \in g(\textit{GP}) }u(\textit{GP}, T_j)$, where $g(\textit{GP})$ is defined as $g(\textit{GP})$ = $\{T_j \in \mathcal{D} \mid v \in T_j  \land \forall g \in \textit{GP} \land \exists v \in \textit{Leaf}(g, \tau)\}$.
\end{definition}
	
\begin{definition}\label{def:Cross-Level_high-utility_itemset}
	\rm \textbf{(Cross-level high-utility itemset) \cite{fournier2020mining}}. A (generalized) itemset $\textit{GP}$ is considered a cross-level high-utility itemset (CLHUI) if and only if the utility of the itemset \textit{GP} is greater than or equal to the minimum utility threshold \textit{minutil}: $u(P) \geq$ \textit{minutil}.
\end{definition}
	
The transaction utility (TU) of a transaction $T_j$ is defined as the sum of utilities of all items in $T_j$, calculated as $\textit{TU}(T_j)$ = $\sum_{x \in T_j} u(x, T_j)$ \cite{liu2005two}. For example, in Fig \ref{taxonomy}, the leaf items of the generalized item $X$ are $\textit{Leaf}(\{X\}, \tau)$ = $\{a, b, c\}$, and its descendant items are $\textit{Desc}(\{X\}, \tau)$ = $\{Y, a, b, c\}$. The level of item $Y$ is $\textit{level}(Y)$ = 2. In Table \ref{Tab:transaction}, $\textit{TU}(T_8)$ = $u(a,T_8)$ + $u(b,T_8)$ + $u(e,T_8)$ = 5 + 4 + 6 = 15. Assuming \textit{minutil} = 50, according to the example in Table \ref{Tab:transaction}, $u(X, T_5)$ = $u(a, T_5)$ + $u(b, T_5)$ = 1 $\times$ 5 + 1 $\times$ 1 = 6.  $u(X)$ = $u(X, T_1)$ + $u(X, T_3)$ + $u(X, T_5)$ + $u(X, T_7)$ + $u(X, T_8)$ = 6 + 10 + 22 + 6 + 13 + 9 = 66. Since $u(X)$ = 66 > 50, ${X}$ is a CLHUI. Let itemset $P$ = $\{Z, b\}$, $u(P, T_5)$ = $u(b, T_5)$ + $u(d, T_5)$ = 1 + 3 = 4. $u(P)$ = $u(P, T_3)$ + $u(P, T_5)$ + $u(P, T_8)$ = 11 + 4 + 10 = 25. All CLHUIs are shown in Table \ref{Tab:CLHUIs}.

\begin{table}
	\centering	\makeatletter\def\@captype{table}\makeatother\setlength{\belowcaptionskip}{10pt}
	\caption{Cross-level high-utility itemsets.}
	\label{Tab:CLHUIs}
	{\normalsize\fontfamily{ptm}\selectfont  
			\begin{tabular}{p{1.5cm} p{1.5cm} p{1.5cm} p{1.5cm}} 
				\hline
				\textbf{Itemset} & \textbf{Utility} & \textbf{Itemset} & \textbf{Utility}\\
				\hline
				$\{X\}$ & 66 & $\{X, Z\}$ & 85  \\
				$\{X, e\}$ & 55 & $\{X, d\}$ & 62  \\
				$\{Z\}$ & 69 & $\{Z, Y\}$ & 70  \\
				$\{Z, a\}$ & 62 & $\{e, d\}$ & 57  \\
				\hline
			\end{tabular}
	}
\end{table}
	
\begin{definition}\label{def:SCLHUIs}
	\rm \textbf{(Sensitive and non-sensitive cross-level high-utility itemsets)}. The set of CLHUIs containing private or sensitive information that needs to be hidden is defined as sensitive CLHUIs (SCLHUIs). Let SCLHUIs = \{$sI_1$, $sI_2$, ..., $sI_m$\}, this set represents the SCLHUIs that are hidden in the database. The set of CLHUIs that need to be preserved to ensure the data availability is defined as the non-sensitive CLHUIs (NSCLHUIs). Let NSCLHUIs = \{$nsI_1$, $nsI_2$, ..., $nsI_k$\}, where SCLHUIs $\cup$ NSCLHUIs = CLHUIs.
\end{definition}
	
\begin{definition}\label{def:SC_T}
	\rm \textbf{(Sensitive count, non-sensitive count, and weight for transactions) \cite{ashraf2023efficient,jangra2022efficient}}. A transaction $T_j$ is a sensitive transaction $ST_j$ if there exists an SCLHUI $s_i$ contained in the transaction $T_j$, and the set of all sensitive transactions is the sensitive transaction set, denoted as $ST$. The number of SCLHUIs appearing in a transaction $T_j$ is denoted as $\textit{SC}(T_j)$. The number of NSCLHUIs appearing in a transaction $T_j$ is denoted as $\textit{NSC}(T_j)$. The sensitive weight of $T_j$ is calculated as follows: $\textit{Wt}(T_j ) = \frac{\textit{SC}(T_j)}{\textit{NSC}(T_j )+1}.$
\end{definition}
	
Let SCLHUIs = $\{\{X, d\}, \{Z, Y\}, \{e, d\}\}$. The transactions in which they appear are shown in Table \ref{Tab:SCLHUIs}. Therefore, the sensitive transaction $ST$ = \{$T_1$, $T_2$, $T_3$, $T_4$, $T_5$, $T_6$, $T_8$\}. Since there are $\{X,d\}$ and $\{Z,Y\}$ in $T_1$, so $\textit{SC}(T_1)$ = 2. Additionally, because $\{X\}$, $\{Z\}$, $\{Z, a\}$ and $\{X, Z\}$ are in $T_1$ so $\textit{NSC}(T_1)$ = 2, $\textit{Wt}(T_1)$ = 2 / ( 4 + 1 ) = 0.4. \textit{SC}, \textit{NSC}, and \textit{Wt} of $ST$ are shown in Table \ref{Tab:SC_T}.
	
\begin{table}[width=0.5\textwidth]
	\centering	\makeatletter\def\@captype{table}\makeatother\setlength{\belowcaptionskip}{10pt}
		\caption{Sensitive cross-level high-utility itemsets and their transactions.}
		\label{Tab:SCLHUIs}
		{\normalsize\fontfamily{ptm}\selectfont  
			\begin{tabular}{p{1.5cm} p{2.5cm}} 
				\hline
				\textbf{SCLHUI} & \textbf{Transactions}\\
				\hline
				$\{X, d\}$ & $T_1$, $T_2$, $T_3$, $T_5$ \\
				$\{Z, Y\}$ & $T_1$, $T_2$, $T_3$, $T_5$, $T_8$  \\
				$\{e, d\}$ & $T_2$, $T_3$, $T_4$, $T_6$  \\
				\hline
			\end{tabular}
		}
\end{table}

\begin{table}
	\centering	\makeatletter\def\@captype{table}\makeatother\setlength{\belowcaptionskip}{10pt}
	\caption{Sensitive counts, non-sensitive counts, and weights for transactions.}
	\label{Tab:SC_T}
	{\normalsize\fontfamily{ptm}\selectfont  
	\begin{tabular}{p{0.7cm} p{0.6cm} p{0.6cm} p{0.6cm} p{0.6cm} p{0.6cm} p{0.6cm} p{0.6cm}} 
		\hline
		\textbf{$ST$} & \textbf{$T_1$} & \textbf{$T_2$} & \textbf{$T_3$} & \textbf{$T_4$} & \textbf{$T_5$} & \textbf{$T_6$} & \textbf{$T_8$}\\
		\hline
		\textit{SC} & 2 & 3 & 3 & 1 & 2 & 1 & 1  \\
		\textit{NSC} & 4 & 5 & 5 & 1 & 4 & 1 & 5  \\
		\textit{Wt} & 0.40 & 0.50 & 0.50 & 0.50 & 0.40 & 0.50 & 0.17  \\
		\hline
		\end{tabular}
	}
\end{table}

There may be no NSCLHUIs in some sensitive transactions; therefore, 1 is added to the denominator during \textit{Wt} calculation. The higher SC of the transaction can affect more SCLHUIs when processing the transaction, speeding up the algorithm's efficiency. The lower NSC of the transaction can reduce the impact on NSCLHUIs, minimizing the algorithm's side effects. By sorting transactions based on their \textit{Wt}, it is able to prioritize the more suitable transactions as victim transactions \cite{ashraf2023efficient}. According to Table \ref{Tab:SC_T}, the sample ordering of $ST$ is $T_2 \prec T_3 \prec T_4 \prec T_6 \prec T_1 \prec T_5 \prec T_8$.

Because generalized items are associated with taxonomy information, hiding them is more complicated than hiding ordinary items. When the utility of an ordinary itemset is reduced by deleting a leaf item from a transaction, that itemset will no longer appear in the transaction once one of its constituent items is removed. In contrast, for a generalized itemset, deleting a leaf item does not necessarily remove the generalized itemset from the transaction, because other leaf items of the same generalized item may still remain. Traditional PPUM algorithms do not distinguish between ordinary items and generalized items, and therefore cannot correctly handle this situation. For example, consider the database in Table \ref{Tab:transaction} and the generalized itemset $\{X,d\}$. If the item $d$ is removed from transaction $T_1$, then $\{X,d\}$ no longer appears in $T_1$, and its utility becomes $u(\{X,d\})$ = $u(\{X,d\})$ - $u(\{X,d\}, T_1)$ = 62 - 5 - 1 - 3 = 53. However, if the item $a$ is removed from $T_1$, the item $b$ still remains in the transaction. Since both $a$ and $b$ belong to the generalized item $X$, the itemset $\{X,d\}$ still appears in $T_1$, and its utility becomes$u(\{X,d\})$ = $u(\{X,d\})$ - $u(a, T_1)$ = 62 - 5 = 57. This difference shows that the hiding mechanism for generalized itemsets cannot be handled in the same way as that for ordinary itemsets.
	
\textbf{Problem statement}: Given a database $\mathcal{D}$, a taxonomy $\tau$, a user-specified minimum utility threshold \textit{minutil}, the set of all cross-level high-utility itemsets CLHUIs, and the set of all sensitive cross-level high-utility itemsets SCLHUIs, the task of the cross-level privacy-preserving utility mining (CLPPUM) is to hide all SCLHUIs in the database $\mathcal{D}$ and output the sanitized database $\mathcal{D'}$. Compared with the traditional PPUM, the key challenge of CLPPUM lies in identifying generalized items that do not explicitly exist in the database and applying reduction or deletion operations on them to achieve the hiding of generalized itemsets.
	
Based on the definition of the CLPPUM, the framework of CLPPUM is illustrated in Fig. \ref{framework}. The heuristic-based item deletion PPUM algorithm selects victim items and victim transactions related to SCLHUI. By modifying or deleting the victim items in the victim transaction, the algorithm reduces the utility of the SCLHUIs below the \textit{minutil}. Common evaluation metrics for PPUM algorithms include hiding failure (HF), missing cost (MC), artificial cost (AC), itemset utility similarity (IUS), database utility similarity (DUS), and transaction modification ratio (TMR).
	
\begin{figure*}[htbp]
	\centering
	\includegraphics[width=0.75\textwidth]{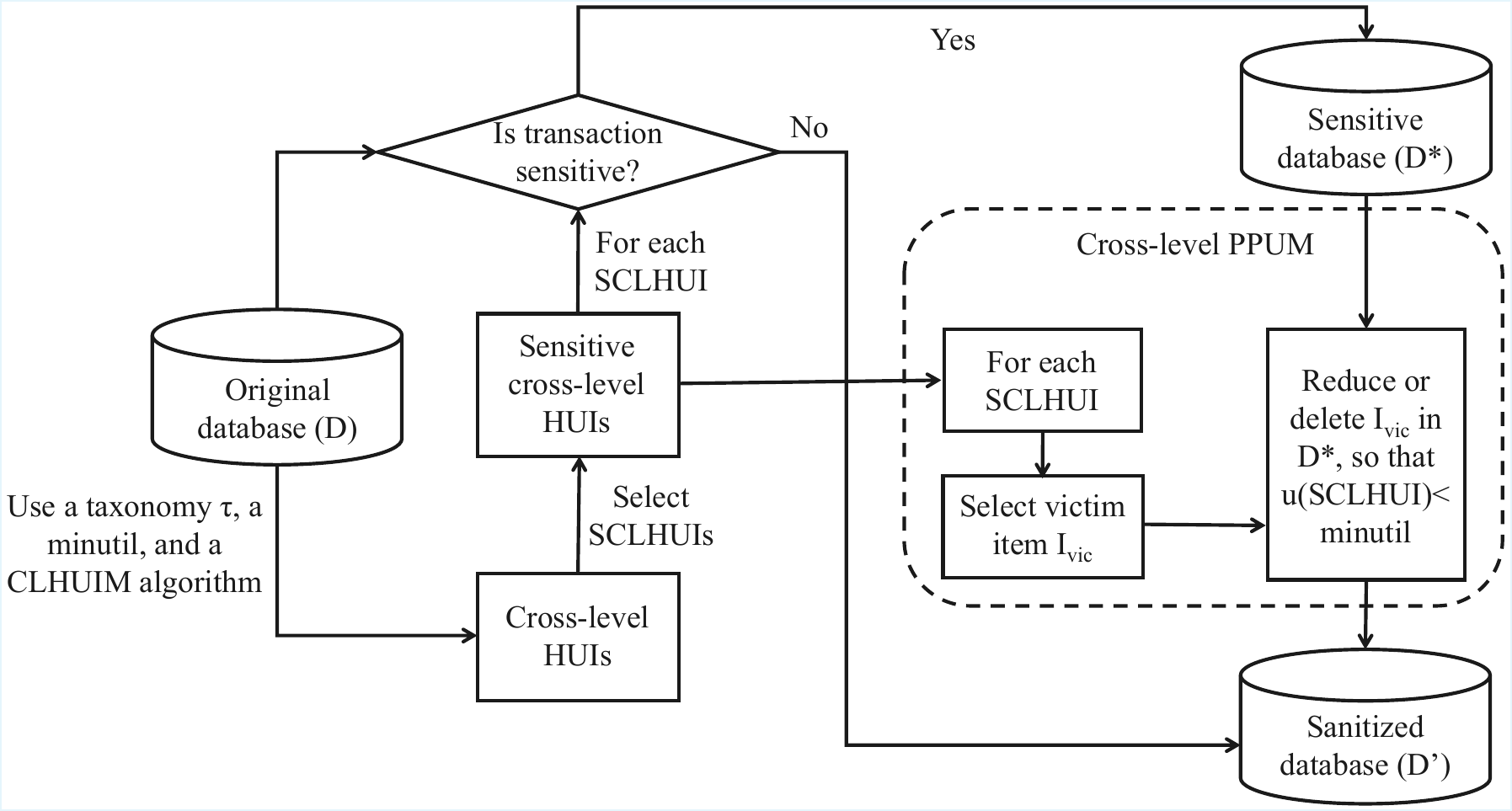} 
	\caption{Algorithmic framework of CLPPUM.}
	\label{framework} 
\end{figure*}
	
\begin{definition}\label{def:HF}
	\rm \textbf{(Hiding failure) \cite{lin2016fast,yeh2010hhuif}}. Hiding failure is defined as the ratio of SCLHUIs that remain unhidden after database sanitization to the total number of SCLHUIs. This ratio is calculated by the following formula: $\textit{HF} = \frac{|\textit{SCLHUIs} \cap \textit{CLHUIs'}|}{|\textit{SCLHUIs}|}.$
\end{definition}
	
\begin{definition}\label{def:MC}
	\rm \textbf{(Missing cost) \cite{lin2016fast,yeh2010hhuif}}. Missing cost is defined as the ratio of NSCLHUIs that are hidden after database sanitization to the total number of NSCLHUIs. This ratio is calculated by the following formula: $\textit{MC} = \frac{|\textit{NSCLHUIs} - \textit{CLHUIs'}|}{|\textit{NSCLHUIs}|}.$
\end{definition}
	
\begin{definition}\label{def:AC}
	\rm \textbf{(Artificial cost) \cite{lin2016fast}}. Artificial cost is defined as the ratio of false CLHUIs that are discovered after database sanitization to the total number of CLHUIs. This ratio is calculated by the following formula: $\textit{AC} = \frac{|\textit{CLHUIs'} - \textit{CLHUIs}|}{|\textit{CLHUIs'}|}.$
\end{definition}
	
The relationship between the above three side effects and the CLHUIs is shown in Fig. \ref{HF}. The upper circle (\textit{CLHUIs}) and the lower circle (\textit{CLHUIs'}) represent the sets of CLHUIs mined from the original database and the sanitized database, respectively. The yellow forward-slashed region denotes \textit{MC}, the red backward-slashed region denotes \textit{HF}, and the green forward-slashed region denotes \textit{AC}.
	
\begin{figure}[htbp]
	\centering
	\includegraphics[width=0.46\textwidth]{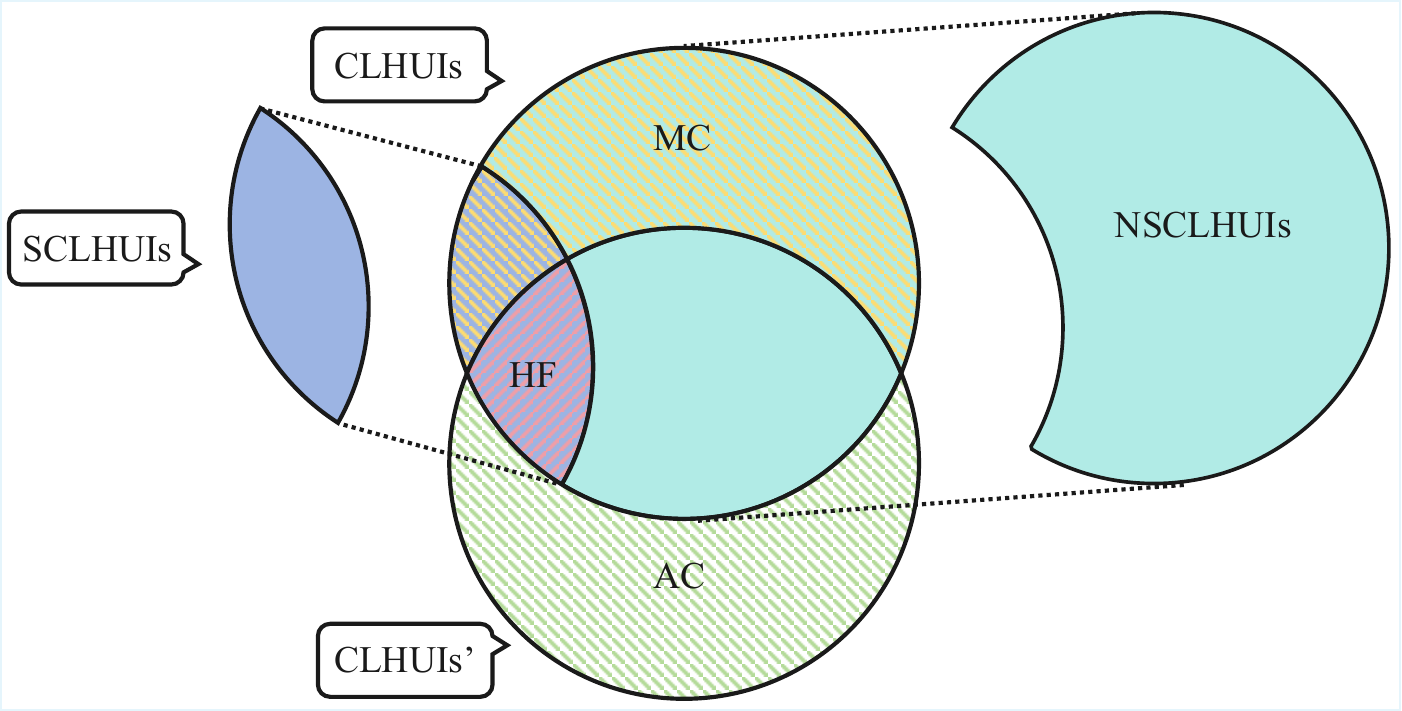} 
	\caption{Relationship of the three side effects to the CLHUIs.}
	\label{HF} 
\end{figure}
	
\begin{definition}\label{def:sim}
	\rm \textbf{(Similarity measures) \cite{lin2016fast}}. Itemset utility similarity is defined as the sum of the utilities of all CLHUIs to the sum of the utilities of CLHUIs: \textit{IUS} = $\frac{\sum_{ P \in \textit{CLHUIs'}} u(P)}{\sum_{ P \in \textit{CLHUIs}}u(P)}.$ It reflects the impact of database sanitization on CLHUIs. Database utility similarity is defined as the ratio of the database utility after sanitization to that before sanitization. It reflects the impact of the sanitization process on the utility of the database: \textit{DUS} = $\frac{\sum_{T_j \in \mathcal{D'}} TU(T_j)}{\sum_{T_j \in \mathcal{D}} TU(T_j)}.$ The transaction modification ratio is defined as the ratio of the number of transactions modified during the database sanitization to the total number of all transactions in the database. It is calculated as follows: $\textit{TMR} = \frac{\# \textit{modified transactions}}{|\mathcal{D}|}.$
\end{definition}

\section{Proposed Algorithms} \label{sec:algorithm}
	
This section presents the proposed CLPPUM framework for hiding sensitive cross-level high-utility itemsets. We first describe the metrics used to evaluate the influence of generalized items on sensitive and non-sensitive itemsets, and then introduce the data structure used to support victim item selection. Based on these components, three algorithms, called Minimum \textit{RGISU}-first (Min-RF), Maximum \textit{RGISU}-first (Max-RF), and Best \textit{NSC}-first (Best-NSCF), are developed. Their privacy-related properties are discussed in Section 4.5.
	
\begin{definition}\label{def:SC_I}
    \rm \textbf{(Sensitive and non-sensitive count of (generalized) items) \cite{liu2020improved}}. Sensitive count \textit{SC}$(g)$ of a (generalized) item $g$ refers to the number of SCLHUIs that contain the (generalized) item $g$, its ancestors, or its descendants. A non-sensitive count \textit{NSC}$(g)$ of a (generalized) item $g$ refers to the number of NSCLHUIs that contain the (generalized) item $g$, its ancestors, or its descendants.
\end{definition}
	
\begin{definition}\label{def:RGISU_I}
    \rm \textbf{(Real generalized item sensitive utility) \cite{ashraf2023efficient}}. Let $g$ be a sensitive (generalized) item. The real generalized item sensitive utility (\textit{RGISU}) of $g$ is: $\textit{RGISU}(g)$ = $\sum_{g \in T_j \land T_j \in ST \land v \in \textit{Leaf}(g, \tau) }u(v, T_j)$
\end{definition}
	
When hiding sensitive itemsets, modifying different items may lead to substantially different side effects. This is mainly due to the fact that sensitive and non-sensitive itemsets often share common items or co-occur in the same transactions. As a result, the choice of victim item directly determines which itemsets will be affected during the sanitization process. In addition, a considerable number of transactions in the database do not contain any sensitive itemsets, and the utilities of items in such transactions are irrelevant to the hiding process. Therefore, selecting victim items solely based on utility computed over the entire database cannot accurately reflect their actual impact on sensitive and non-sensitive itemsets. To better capture this impact, we adopt \textit{RGISU} as the primary criterion for victim item selection. Since \textit{RGISU} is computed only from sensitive transactions, it more effectively reflects the contribution of a (generalized) item to the sensitive itemsets that need to be hidden. As generalized items do not explicitly appear in the database, operations on generalized items are transformed into operations on their corresponding leaf items. In this way, the proposed approach is able to preserve the hierarchical structure while effectively hiding SCLHUIs.
	
According to Table \ref{Tab:CLHUIs} and Table \ref{Tab:SCLHUIs}, $\textit{SC}(Y)$ = 2 and $\textit{NSC}(Y)$ = 4. $\textit{RGISU}(Y)$ = $u(Y, T_1)$ + $u(Y, T_2)$ + $u(Y, T_3)$ + $u(Y, T_5)$ + $u(Y, T_8)$ = 6 + 10 + 7 + 6 + 9 = 38. \textit{SC} and \textit{NSC} are used to select items that appear more frequently in sensitive itemsets and less frequently in non-sensitive itemsets, thereby reducing the impact on non-sensitive itemsets while speeding up the hiding of sensitive itemsets. Since items possess hierarchical information, operations on their descendants or ancestors are likely to impact them. Additionally, this counting method enables lower-level nodes to have higher \textit{NSC}. This causes the algorithm to select higher-level nodes with fewer leaf nodes as victim items. These nodes generally have fewer leaf nodes, which reduces the number of deletion operations and speeds up the hiding process for generalized items. The \textit{RGISU} is used to identify the items with the highest or lowest utility in the $ST$. By selecting such an item as the victim item, the algorithm can either accelerate the hiding process or minimize the impact on non-sensitive itemsets.
	
To explain the algorithm more clearly, this section first introduces the proposed \textit{GI-dic} dictionary structure, followed by a description of the three CLPPUM algorithms: Min-RF, Max-RF, and Best-NSCF.
	
\subsection{Constructing GI-dic dictionary}
	
In order to efficiently calculate and use the assessment metrics to select the victim item $I_{vic}$, this paper constructs a generalized item-dictionary \textit{GI-dic}, which can effectively reduce the traversal consumption of $\mathcal{D}$, SCLHUIs, and NSCLHUIs in the process of calculation.
	
\begin{definition}\label{def:GI-dic}
    \rm \textbf{(Generalized item-dictionary)}. Let $g$ be a (generalized) item, and \textit{GI-dic} stores the transactions $T_j$ in $g$ appears, along with the metrics used for selecting victim items.  \textit{GI-dic} takes $g$ as the key, and the values are \textit{SC}, \textit{NSC}, and \textit{RGISU} of $g$ and the $ST$ in which $g$ appears ($\textit{GI-dic}$[$G_i$] = \{\textit{SC}, \textit{NSC}, \textit{RGISU}, \textit{STs}\}). 
\end{definition}
	
Since generalized items do not exist in the database, additional checks are required to identify whether an itemset is present in the current transaction. When computing the \textit{SC} and \textit{NSC} of a transaction, it is necessary to traverse and identify all SCLHUIs and NSCLHUIs while traversing the database $\mathcal{D}$, resulting in a time complexity of $O(|\mathcal{D}|*|\textit{NSCLHUIs}|*|\textit{NSCLHUI}|)$. \textit{GI-dic} enables fast identification of the transactions containing generalized itemsets through XOR operations. This optimization allows the algorithm to traverse the $\mathcal{D}$, SCLHUIs, and the NSCLHUIs once when calculating the \textit{SC} and \textit{NSC} of a transaction, reducing the time complexity to $O(|\textit{NSCLHUIs}|*|\textit{NSCLHUI}|*|\mathcal{D}_{\textit{NSI}}|)$, where $\mathcal{D}_{\textit{NSI}}$ denotes the set of transactions in which items of NSCLHUI appear. This optimization can improve the performance of the algorithm by 30-50\%. In addition, \textit{GI-dic} can store various metrics of (generalized) items, assisting in the selection of victim items. Based on the introduction of the above definition, the process of constructing \textit{GI-dic} dictionary is shown in Algorithm \ref{algo:GI-dic}.
	
The construction process in Algorithm \ref{algo:GI-dic} is as follows. First, \textit{GI-dic} is initialized for all items (lines 1–4). Next, scan the database once to compute \textit{RGISU} for the (generalized) item $G_i$ in the sensitive transaction and record the TID of this transaction into \textit{GI-dic} (lines 5-15). Specifically, the algorithm first determines whether the current transaction contains any SCLHUI, i.e., whether the transaction is a $ST$ (lines 6–7). If so, the algorithm iterates through all (generalized) items appearing in the transaction, computes their \textit{RGISU}, and records both the \textit{RGISU} and the TIDs of the $ST$ in \textit{GI-dic} (lines 8–11). After that, the algorithm processes each SCLHUI $S_k$ in sequence (lines 16–26). It first traverses all (generalized) items $G_i$ in $S_k$ to compute their \textit{SC}, and uses \textit{GI-dic} to identify all transactions containing $S_k$ (lines 17–22). For each such transaction $T_j$, the corresponding \textit{SC} is increased (lines 23–25). The algorithm then processes the NSCLHUIs $\textit{NS}_k$ (lines 27–37). For each $\textit{NS}_k$, all (generalized) items $G_i$ are traversed to compute their \textit{NSC}, and \textit{GI-dic} is used to locate the transactions in which $\textit{NS}_k$ appears (lines 28–33). The \textit{NSC} of these transactions $T_j$ is then updated accordingly (lines 34–36). Finally, the weight $\textit{Wt}$ of each transaction is calculated, and all transactions are sorted in descending order of $\textit{Wt}$ (line 38). A sample \textit{GI-dic} is shown in Table \ref{Tab:GI-dic}.

\begin{algorithm}[h]
	\small
	\caption{\textbf{GI-dic} construction algorithm}
	\label{algo:GI-dic}
	\LinesNumbered
	\KwIn{a quantitative database $\mathcal{D}$, a taxonomy $\tau$, sensitive cross-level high-utility itemsets \textit{SCLHUIs}, non-sensitive cross-level high-utility itemsets \textit{NSCLHUIs}.}
	\KwOut{generalized item-dictionary $\textit{GI-dic}$.}
		
	\textit{GI-dic} = $\varnothing$\;
		
	\For{\rm \textbf{each} (generalized) item $G_i$ $\in$ $AI$}{
		$\textit{GI-dic}$[$G_i$] = \{0, 0, 0, $\varnothing$\}\;
	}
		
	\For{\rm \textbf{each} transaction $T_j \in \mathcal{D}$}{
			\For{\rm \textbf{each} SCLHUI $S_k \in \textit{SCLHUIs}$}{
				\If {$S_k$ $\subseteq$ $T_j$}{
					\For{\rm \textbf{each} (generalized) item $G_i$ $\in$ $T_j$}{
						$\textit{GI-dic}$[$G_i$].\textit{RGISU} += $u(G_i, T_j)$\;
						$\textit{GI-dic}$[$G_i$].\textit{STs}.\textit{append(j)}\;
					}
					break\;
				}
			}
		}
		
		\For{\rm \textbf{each} SCLHUI $S_k \in \textit{SCLHUIs}$}{
			\textit{Ts} = $\textit{GI-dic}$[$G_i$].\textit{STs} $\land$ $G_i$ $\in$ $S_k$\;
			\For{\rm \textbf{each} (generalized) item $G_i$ $\in$ $S_k$}{
				$\textit{GI-dic}$[$G_i$].sc++\;
				Update the \textit{SC} of descendants and ancestors of $G_i$ according to the taxonomy $\tau$\;
				\textit{Ts} = \textit{Ts} $\cap$ $\textit{GI-dic}$[$G_i$].\textit{STs}\;
			}
			\For{\rm \textbf{each} transaction $T_j \in \textit{transactions}$}{
				$T_j$.sc++\;
			}
		}
		
		\For{\rm \textbf{each} NSCLHUI $\textit{NS}_k \in \textit{NSCLHUIs}$}{
			\textit{Ts} = $\textit{GI-dic}$[$G_i$].\textit{STs} $\land$ $G_i$ $\in$ $\textit{NS}_k$\;
			\For{\rm \textbf{each} (generalized) item $G_i$ $\in$ $\textit{NS}_k$}{
				$\textit{GI-dic}$[$G_i$].nsc++\;
				Update the \textit{NSC} of descendants and ancestors of $G_i$ according to the taxonomy $\tau$\;
				\textit{Ts} = \textit{Ts} $\cap$ $\textit{GI-dic}$[$G_i$].\textit{STs}\;
			}
			\For{\rm \textbf{each} transaction $T_j \in \textit{transactions}$}{
				$T_j$.nsc++\;
			}
		}
		
	Sort transactions $T_j \in \mathcal{D}$ according to their \textit{Wt} in descending order, where \textit{Wt} = $T_j.sc$ / ($T_j$.\textit{nsc} + 1 )\;
	\Return generalized item-dictionary $\textit{GI-dic}$\;	
\end{algorithm}
	
\begin{table}
	\centering	\makeatletter\def\@captype{table}\makeatother\setlength{\belowcaptionskip}{10pt}
	\caption{GI-dic example.}
	\label{Tab:GI-dic}
	{\normalsize\fontfamily{ptm}\selectfont  
			\begin{tabular}{p{0.7cm} p{0.5cm} p{0.7cm} p{1.1cm} p{3.2cm}} 
				\hline
				\textbf{Item} & \textbf{$\textit{SC}$} & \textbf{$\textit{NSC}$} & \textbf{$\textit{RGISU}$} & \textbf{Transactions} \\
				\hline
				$a$ & 2 & 4 & 30 & $T_1, T_2, T_3, T_5, T_8$\\
				$b$ & 2 & 3 & 8 & $T_1, T_3, T_5, T_8$\\
				$c$ & 1 & 3 & 15 & $T_3, T_7$\\
				$d$ & 3 & 3 & 45 & $T_1, T_2, T_3, T_4, T_5, T_6$\\
				$e$ & 2 & 4 & 24 & $T_2, T_3, T_4, T_6, T_8$\\
				$f$ & 0 & 0 & 2 & $T_6$\\
				$X$ & 2 & 4 & 53 & $T_1, T_2, T_3, T_5, T_8$\\
				$Y$ & 2 & 4 & 38 & $T_1, T_2, T_3, T_5, T_8$\\
				$Z$ & 3 & 4 & 69 & $T_1, T_2, T_3, T_4, T_5, T_6, T_8$\\
				\hline
			\end{tabular}
		}
	\end{table}
	
\subsection{Min-RF algorithm}

This section introduces the proposed minimum \textit{RGISU}-first algorithm (Min-RF). The algorithm selects the (generalized) item $I_i$ with the smallest \textit{RGISU} as the victim item, aiming to minimize the impact on NSCLHUIs. The pseudo-code of the Min-RF algorithm is shown in Algorithm \ref{algo:Min-RF}.
	
First, the algorithm scans the database once and calculates and stores the utility of generalized items (line 1). Next, it calls \textit{GI-dic} construction algorithm (Algorithm \ref{algo:GI-dic}) to calculate various metrics for (generalized) item $G_i$ and transaction $T_j$, and sorts transactions $T_j$ in descending order by weight (line 2). Then, it iterates over all SCLHUIs, selects the (generalized) item $I_i$ with the smallest \textit{RGISU} as the victim item $I_{vic}$, and sorts all SCLHUIs $S_i$ in descending order based on the \textit{RGISU} of its $I_{vic}$ (lines 3–6). Finally, the algorithm iterates through all SCLHUI $S_i$ and hides them (lines 7–34). The algorithm first calculates the utility reduction \textit{diff} required to hide $S_i$, then obtains the leaf items $I_{vics}$ of $I_{vic}$, and sorts $I_{vics}$ in descending order based on their \textit{RGISU}, transforming the processing of the (generalized) items into the processing of their leaf items (lines 8–9). Then, traverse $T_j$ contains the current SCLHUI $S_i$, and processes the victim item $I_{vic}$ in the transaction until \textit{diff} $\leq$ 0 and $S_i$ is successfully hidden (lines 10-33). The algorithm traverses the leaf items. If \textit{diff} $\geq$ $u(LI, T_j)$, $LI$ is removed from $T_j$ (lines 15–23). If $u(LI, T_j)$ = $u(I_{vic}, T_j)$, it means that deleting $LI$ will also remove $I_{vic}$ from $T_j$, and $S_i$ will no longer be in $T_j$. In this case, \textit{diff} reduces the utility of $S_i$ in $T_j$. Otherwise, \textit{diff} only reduces the utility of $LI$ in $T_j$ (lines 16–20). The algorithm updates the impact of removing $LI$ on all SCLHUI $S_k$ and removes $LI$ from $T_j$ (lines 21–22). If \textit{diff} $\leq$ $u(LI, T_j)$, the internal utility of $LI$ in $T_j$ is reduced. The algorithm calculates the required decrease in internal utility (\textit{diu}), updates the database for all SCLHUI $S_k$, setting diff to 0 (lines 24–28). Finally, the algorithm completes the hiding of SCLHUIs and returns the sanitized database $\mathcal{D'}$ (line 35).
	
\begin{algorithm}[!ht]
	\small
	\caption{\textbf{The Min-RF algorithm}}
	\label{algo:Min-RF}
	\LinesNumbered
	\KwIn{a quantitative database $\mathcal{D}$, a taxonomy $\tau$, minimum utility threshold \textit{minutil}, sensitive cross-level high-utility itemsets \textit{SCLHUIs}, cross-level high-utility itemsets \textit{CLHUIs}.}
	\KwOut{A sanitized database $\mathcal{D'}$.}
		
	Scan the database $\mathcal{D}$ to compute and store the utility of generalized items\;
		
	Calls \textbf{GI-dic construction algorithm} to compute the \textit{SC}, \textit{NSC}, and \textit{RGISU} of (generalized) items, and to compute the \textit{SC}, \textit{NSC}, and \textit{Wt} of transactions to construct the \textbf{GI-dic} dictionary.
		
		\For{\rm \textbf{each} SCLHUI $S_i \in \textit{SCLHUIs}$}{
			$I_{vic}(S_i)$ = $I_i$, where $I_i \in S_i$ $\land$ $\forall I_j \in S_i$, $\textit{RGISU}(I_i) \leq \textit{RGISU}(I_j)$\;
		}
		
		Sort the SCLHUI $S_i$ $\in$ \textit{SCLHUIs} in descending order according to the \textit{RGISU} of $I_{vic}(S_i)$\;
		
		\For{\rm \textbf{each} SCLHUI $S_i \in \textit{SCLHUIs}$}{
			
			\textit{diff} = $u(S_i) - \textit{minutil} + 1$\;
			
			Obtain all leaf items $I_{vics}$ of $I_{vics}$, and sort $I_{vics}$ in ascending order according to their \textit{RGISU}\;
			
			\For{\rm \textbf{each}transaction $T_j \in \mathcal{D}$}{
				\If {\textit{diff} > 0 $\land$ $S_i \subseteq T_j$}{
					$T_{vic}(S_i)$ = $T_j$\;
					\For{\rm \textbf{each} leaf item $LI \in I_{vics}$}{
						\If {\textit{diff} > 0 $\land$ $LI \in T_j$}{
							\uIf {\textit{diff} $\geq$ $u(LI, T_j)$}{
								\uIf {$u(LI, T_j)$ == $u(I_{vic}, T_j)$}{
									\textit{diff} -= $u(S_i, T_j)$\;
								}\uElse{
									\textit{diff} -= $u(LI, T_j)$\;
								}{\textbf{end}}
								
								Update the SCLHUI $S_k$ $\subseteq$ $\textit{SCLHUIs}$, where $S_k \subseteq T_j$ $\land$ ($LI \in S_k$ $\lor$ $\exists d \in S_k \land LI \in \textit{Desc}(d,\tau)$)\;
								Remove $LI$ from $T_j$\;

							}\uElse{
								\textit{diu} = $\lceil \textit{diff}/eu(I_{vic}) \rceil$\;
								$\textit{iu}(\textit{LI}, T_j)$ -= $\textit{diu}$\;
								Update the SCLHUI $S_k$ $\subseteq$ $\textit{SCLHUIs}$, where $S_k \subseteq T_j$ $\land$ ($LI \in S_k$ $\lor$ $\exists d \in S_k \land LI \in \textit{Desc}(d,\tau)$)\;
								Update $T_j$\;
								\textit{diff} = 0\;
							}{\textbf{end}}
						}		
					}
				}
				
			}
		}
		\Return a sanitized database $\mathcal{D'}$\;		
	\end{algorithm}	
\subsection{Max-RF algorithm}
	
The difference between the maximum \textit{RGISU}-first algorithm (Max-RF) and Min-RF lies in the selection strategy: Max-RF selects the items with the highest \textit{RGISU} as victim items, aiming to reduce the number of deletions involving victim items, thereby improving overall runtime efficiency (line 4). In addition, when sorting the leaf items of $I_{vic}$, the algorithm sorts them in descending order according to \textit{RGISU}, also to minimize the number of operations and improve efficiency (line 6). The pseudocode of the Max-RF algorithm is shown in Algorithm \ref{algo:Max-RF}.

\begin{algorithm}[h!]
	\small
	\caption{\textbf{The Max-RF algorithm}}
	\label{algo:Max-RF}
	\LinesNumbered
	\KwIn{a quantitative database $\mathcal{D}$, a taxonomy $\tau$, minimum utility threshold \textit{minutil}, sensitive cross-level high-utility itemsets \textit{SCLHUIs}, cross-level high-utility itemsets \textit{CLHUIs}.}
	\KwOut{A sanitized database $\mathcal{D'}$.}
		
		Scan the database $\mathcal{D}$ to compute and store the utility of generalized items\;
		
		Calls \textbf{GI-dic construction algorithm} to compute the \textit{SC}, \textit{NSC}, and \textit{RGISU} of (generalized) items, and to compute the \textit{SC}, \textit{NSC}, and \textit{Wt} of transactions to construct the \textbf{GI-dic} dictionary.
		
		\For{\rm \textbf{each} SCLHUI $S_i \in \textit{SCLHUIs}$}{
			$I_{vic}(S_i)$ = $I_i$, where $I_i \in S_i$ $\land$ $\forall I_j \in S_i$, $\textit{RGISU}(I_i) \geq \textit{RGISU}(I_j)$\;
		}
		
		Sort the SCLHUI $S_i$ $\in$ \textit{SCLHUIs} in descending order according to the \textit{RGISU} of $I_{vic}(S_i)$\;
		
		\For{\rm \textbf{each} SCLHUI $S_i \in \textit{SCLHUIs}$}{
			
			\textit{diff} = $u(S_i) - \textit{minutil} + 1$\;
			
			Obtain all leaf items $I_{vics}$ of $I_{vics}$, and sort $I_{vics}$ in descending order according to their \textit{RGISU}\;
			
			\For{\rm \textbf{each} transaction $T_j \in \mathcal{D}$}{
				\If {\textit{diff} > 0 $\land$ $S_i \subseteq T_j$}{
					$T_{vic}(S_i)$ = $T_j$\;
					\For{\rm \textbf{each} leaf item $LI \in I_{vics}$}{
						\If {\textit{diff} > 0 $\land$ $LI \in T_j$}{
							\uIf {\textit{diff} $\geq$ $u(LI, T_j)$}{
								\uIf {$u(LI, T_j)$ == $u(I_{vic}, T_j)$}{
									\textit{diff} -= $u(S_i, T_j)$\;
								}\uElse{
									\textit{diff} -= $u(LI, T_j)$\;
								}{\textbf{end}}
								
								Update the SCLHUI $S_k$ $\subseteq$ $\textit{SCLHUIs}$, where $S_k \subseteq T_j$ $\land$ ($LI \in S_k$ $\lor$ $\exists d \in S_k \land LI \in \textit{Desc}(d,\tau)$)\;
								Remove $LI$ from $T_j$\;
							}\uElse{
								\textit{diu} = $\lceil \textit{diff}/eu(I_{vic}) \rceil$\;
								$\textit{iu}(\textit{LI}, T_j)$ -= $\textit{diu}$\;
								Update the SCLHUI $S_k$ $\subseteq$ $\textit{SCLHUIs}$, where $S_k \subseteq T_j$ $\land$ ($LI \in S_k$ $\lor$ $\exists d \in S_k \land LI \in \textit{Desc}(d,\tau)$)\;
								Update $T_j$\;
								\textit{diff} = 0\;
							}{\textbf{end}}
						}		
					}
				}
				
			}
		}
		\Return a sanitized database $\mathcal{D'}$\;
\end{algorithm}
	
\begin{algorithm}[h!]
	\small
	\caption{\textbf{The Best-NSCF algorithm}}
	\label{algo:Best-NSCF}
	\LinesNumbered
	\KwIn{a quantitative database $\mathcal{D}$, a taxonomy $\tau$, minimum utility threshold \textit{minutil}, sensitive cross-level high-utility itemsets \textit{SCLHUIs}, cross-level high-utility itemsets \textit{CLHUIs}.}
	\KwOut{A sanitized database $\mathcal{D'}$.}
		
		Scan the database $\mathcal{D}$ to compute and store the utility of generalized items\;
		
		Calls \textbf{GI-dic construction algorithm} to compute the \textit{SC}, \textit{NSC}, and \textit{RGISU} of (generalized) items, and to compute the \textit{SC}, \textit{NSC}, and \textit{Wt} of transactions to construct the \textbf{GI-dic} dictionary.
		
		\For{\rm \textbf{each} SCLHUI $S_i \in \textit{SCLHUIs}$}{
			$I_{vic}(S_i)$ = $I_i$, where $I_i \in S_i$ $\land$ $\forall I_j \in S_i$, $\textit{NSC}(I_i) \leq \textit{NSC}(I_j)$ $\land$ $\textit{SC}(I_i) \geq \textit{SC}(I_j)$, if no such item exists, then select $I_i \in S_i$ $\land$ $\forall I_j \in S_i$, $\textit{NSC}(I_i) \leq \textit{NSC}(I_j)$ $\land$ $\textit{RGISU}(I_i) \leq \textit{RGISU}(I_j)$;
		}
		
		Sort the SCLHUI $S_i$ $\in$ \textit{SCLHUIs} in descending order according to the \textit{RGISU} of $I_{vic}(S_i)$\;
		
		\For{\rm \textbf{each} SCLHUI $S_i \in \textit{SCLHUIs}$}{
			
			\textit{diff} = $u(S_i) - \textit{minutil} + 1$\;
			
			Obtain all leaf items $I_{vics}$ of $I_{vics}$, and sort $I_{vics}$ in ascending order according to their \textit{RGISU}\;
			
			\For{\rm \textbf{each} transaction $T_j \in \mathcal{D}$}{
				\If {\textit{diff} > 0 $\land$ $S_i \subseteq T_j$}{
					$T_{vic}(S_i)$ = $T_j$\;
					\For{\rm \textbf{each} leaf item $LI \in I_{vics}$}{
						\If {\textit{diff} > 0 $\land$ $LI \in T_j$}{
							\uIf {\textit{diff} $\geq$ $u(LI, T_j)$}{
								\uIf {$u(LI, T_j)$ == $u(I_{vic}, T_j)$}{
									\textit{diff} -= $u(S_i, T_j)$\;
								}\uElse{
									\textit{diff} -= $u(LI, T_j)$\;
								}{\textbf{end}}
								
								Update the SCLHUI $S_k$ $\subseteq$ $\textit{SCLHUIs}$, where $S_k \subseteq T_j$ $\land$ ($LI \in S_k$ $\lor$ $\exists d \in S_k \land LI \in \textit{Desc}(d,\tau)$)\;
								Remove $LI$ from $T_j$\;
							}\uElse{
								\textit{diu} = $\lceil \textit{diff}/eu(I_{vic}) \rceil$\;
								$\textit{iu}(\textit{LI}, T_j)$ -= $\textit{diu}$\;
								Update the SCLHUI $S_k$ $\subseteq$ $\textit{SCLHUIs}$, where $S_k \subseteq T_j$ $\land$ ($LI \in S_k$ $\lor$ $\exists d \in S_k \land LI \in \textit{Desc}(d,\tau)$)\;
								Update $T_j$\;
								\textit{diff} = 0\;
							}{\textbf{end}}
						}		
					}
				}
				
			}
		}
		\Return a sanitized database $\mathcal{D'}$\;
\end{algorithm}
	
\subsection{Best-NSCF algorithm}
	
The difference between the Best \textit{NSC}-first algorithm (Best-NSCF) and Min-RF is the selection strategy. Best-NSCF aims to select the item that appears most frequently in SCLHUIs and occurs least in NSCLHUIs as the victim item $I_{vic}$. When no such optimal item exists, the algorithm selects the items with the lowest \textit{NSC} and lowest \textit{RGISU} as the victim items to minimize the impact on NSCLHUIs (line 4). The pseudocode of the Best-NSCF algorithm is shown in Algorithm \ref{algo:Best-NSCF}.
		
\subsection{Security analysis}

This section analyzes the privacy-related properties of the proposed CLPPUM algorithm from two aspects: the hiding of sensitive patterns and the difficulty of inferring original information from the sanitized database. We focus on how the proposed modifications affect the observability of SCLHUIs and the inference process based on $\mathcal{D'}$.

\begin{definition}\label{def:Objective}
	\rm \textbf{(Privacy Preservation Objective)}. 	The privacy objective of CLPPUM is defined as follows. In the sanitized database $\mathcal{D'}$, for any $P \in \textit{SCLHUIs}$, $u(P, \mathcal{D'}) < \textit{minutil}$, so that $P$ is no longer identified as a CLHUI. Meanwhile, given only the sanitized database $\mathcal{D'}$, \textit{minutil}, and the publicly known algorithm, an attacker without access to the original database or additional background knowledge cannot uniquely determine the original sensitive itemsets.
\end{definition}

\begin{theorem}\label{the:Hiding}
	\rm \textbf{(hiding property)}. For any $P \in \textit{SCLHUIs}$, CLPPUM reduces its utility in the sanitized database $\mathcal{D'}$ such that $u(P,\mathcal{D'}) < \textit{minutil}$.
\end{theorem}
\begin{proof}
	\rm  Let $P$ be a SCLHUI in the original database $\mathcal{D}$, i.e., $u(P,\mathcal{D}) \ge \textit{minutil}$. During the sanitization process, CLPPUM selects victim items and corresponding transactions, and then performs deletion or utility reduction operations. Each modification decreases the utility of $P$. The process continues until the accumulated utility of $P$ falls below the threshold $\textit{minutil}$. According to the definition of CLHUI, $P$ is no longer identified as a CLHUI in $\mathcal{D'}$. Moreover, this reduction is achieved through a sequence of modifications rather than a single deterministic change, which further affects how the modification process can be interpreted from the sanitized data.
\end{proof}

\begin{theorem}\label{the:Perturbation}
	\rm \textbf{(utility structure perturbation)}.
	The sanitization process alters the original utility relationships among items and transactions, making the relationship between observed modifications and original sensitive itemsets less direct.
\end{theorem}
\begin{proof}
	\rm   CLPPUM modifies selected items in victim transactions to reduce the utilities of sensitive itemsets. Since different itemsets may share items or co-occur in the same transactions, modifying a single item can affect multiple itemsets simultaneously. As a result, the mapping between item-level modifications and the affected itemsets is no longer one-to-one. The sanitized database $\mathcal{D'}$ only reflects the final modified utilities, without indicating which sensitive itemset each modification was intended to hide. Therefore, multiple possible original configurations may lead to the same sanitized database $\mathcal{D'}$, making it difficult to uniquely determine the original sensitive itemsets.
\end{proof}

\begin{theorem}\label{the:Expansion}
	\rm \textbf{(expansion of candidate inference space)}.
	After sanitization, sensitive itemsets are mixed with non-sensitive low-utility itemsets, increasing the number of candidates that need to be considered during inference.
\end{theorem}
\begin{proof}
	\rm  After sanitization, for any $P \in \textit{SCLHUIs}$, $u(P, \mathcal{D'})$ $<$ \textit{minutil}. Therefore, they are no longer distinguishable from other low-utility itemsets based on the utility threshold. If an attacker attempts to recover hidden patterns by lowering \textit{minutil}, the number of candidate itemsets increases significantly. In this enlarged candidate set, sensitive itemsets are indistinguishable from many non-sensitive ones. As a result, the attacker examines a much larger set of candidates, which makes accurate identification more difficult in practice.
\end{proof}

In summary, CLPPUM reduces the utilities of sensitive itemsets so that they are no longer identifiable under the given threshold. The modifications introduce dependencies among itemsets and enlarge the set of potential candidates that could explain the observed data. These factors introduce ambiguity into the relationship between the sanitized database and the original sensitive itemsets. As a result, different original databases may correspond to the same sanitized outcome, which limits the ability to uniquely infer sensitive patterns from $\mathcal{D'}$ without additional information.		

\begin{table*}[width=0.8\textwidth,pos=h!]
	\centering	\arraybackslash\makeatletter\def\@captype{table}\makeatother
	\setlength{\belowcaptionskip}{10pt}
	\renewcommand{\arraystretch}{1.2}
	\caption{Characteristics of datasets.}
	\label{Tab:datasets}
		{\scriptsize
			\fontfamily{ptm}\selectfont  
			\begin{tabular}{p{1.3cm} p{1.3cm} p{1.3cm} p{1.3cm} p{1.3cm} p{1.3cm} p{1.3cm} p{1.3cm}}
				\hline
				Dataset & |$\mathcal{D}$| & |$I$| & |$GI$| & Maxlevel & |$T_\textit{MAX}$| & |$T_\textit{AVE}$| & Density\\
				\hline		
				Foodmart & 53,537 & 1,560 & 102 & 5 & 28 & 4.60 & 0.29\% \\
				Fruithut & 181,970 & 1,265 & 43 & 4 & 36 & 3.58 & 0.28\% \\
				Chainstore & 1,112,949 & 40,086 & 11,936 & 10 & 170 & 7.20 & 0.02\% \\
				Chess & 3,196 & 75 & 30 & 3 & 37 & 37.00 & 49.33\% \\
				\hline
			\end{tabular}
		}
\end{table*}

\section{Performance Evaluation}  \label{sec:experiment}
	
To evaluate the performance of the CLPPUM algorithm, all experiments were conducted using Java on a computer equipped with an AMD Ryzen 9 5900HX processor (8 cores) and 16 GB of RAM. To identify a more effective victim item selection strategy, we conduct a comparative analysis of the three proposed algorithms: Min-RF, Max-RF, and Best-NSCF. Traditional PPUM methods cannot effectively analyze or identify generalized itemsets without introducing additional structures. Therefore, HHUIF and MSICF \cite{yeh2010hhuif} were selected as baseline methods in our experiments. These two algorithms have relatively simple structures and can be readily extended to support the identification of generalized items. In the experiments, both the minimum utility threshold (\textit{minutil}) and the number of sensitive itemsets were varied to assess their impact on algorithm performance. Furthermore, the algorithms were evaluated based on seven metrics: runtime, hiding failure (HF), missing cost (MC), artificial cost (AC), itemset utility similarity (IUS), dataset utility similarity (DUS), and transaction modification ratio (TMR).

\subsection{Datasets}
	
We used four datasets containing the taxonomy $\tau$ to evaluate the performance of the algorithms. Foodmart is available in the FEACP repository on GitHub\footnote{Source:\url{https://github.com/nguyenthanhtunghutechsg/FEACP-Evaluation}}, while the Fruithut, Chainstore, and Chess datasets are obtained from the SPMF\footnote{Source:\url{https://www.philippe-fournier-viger.com/spmf/index.php?link=datasets.php}} website. The internal and external utilities of the datasets are generated for the four datasets by taking the largest common factor. Detailed dataset information is provided in Table \ref{Tab:datasets}. $|\mathcal{D}|$ denotes the number of transactions in the dataset $\mathcal{D}$. $|I|$ and $|GI|$ represent the number of primitive and generalized items, respectively. \textit{Maxlevel} indicates the maximum number of levels in the classification hierarchy. $|T_\textit{MAX}|$ and $|T_\textit{AVE}|$ denote the maximum and average number of items contained in each transaction, respectively. \textit{Density} refers to the density of the dataset. Foodmart, Fruithut, and Chainstore are sparse datasets. Chainstore contains the most transactions, (generalized) items, and levels. Chess is the smallest but densest dataset; it has the longest average transaction length. On larger datasets, HHUIF and MSICF may fail to produce results \cite{ashraf2023efficient}. This issue becomes more pronounced in datasets with taxonomy, since the algorithms must handle not only leaf items but also their corresponding generalized items. Therefore, two additional datasets were constructed by extracting the first 5,000 and 10,000 transactions from Foodmart for comparison with the baseline methods.

Since all the proposed algorithms are based on item deletion and can successfully hide all SCLHUIs without introducing false CLHUIs, both the HF and AC of the proposed algorithms are zero. Therefore, this paper does not present the comparison results of these two metrics. When the number of sensitive itemsets is fixed, the number of sensitive itemsets for Chess, Foodmart, Foodmart\_5000, Foodmart\_10000, Fruithut, and Chainstore are: 100, 50, 50, 50, 50, and 2, respectively. When \textit{minutil} is fixed, the \textit{minutil} for Chess, Foodmart, Foodmart\_5000, Foodmart\_10000, Fruithut, and Chainstore are 1,560,000, 580,000, 30,000, 60,000, 3,600,000, and 800,000,000, respectively. Since SCLHUIs typically account for only a small portion of all CLHUIs in real-world scenarios, the proportion of SCLHUIs in the experiments was set to vary between 1\% and 10\% of the total CLHUIs. All sensitive itemsets were randomly selected.

\begin{figure*}[htbp]
	\centering
	\includegraphics[width=0.8\textwidth]{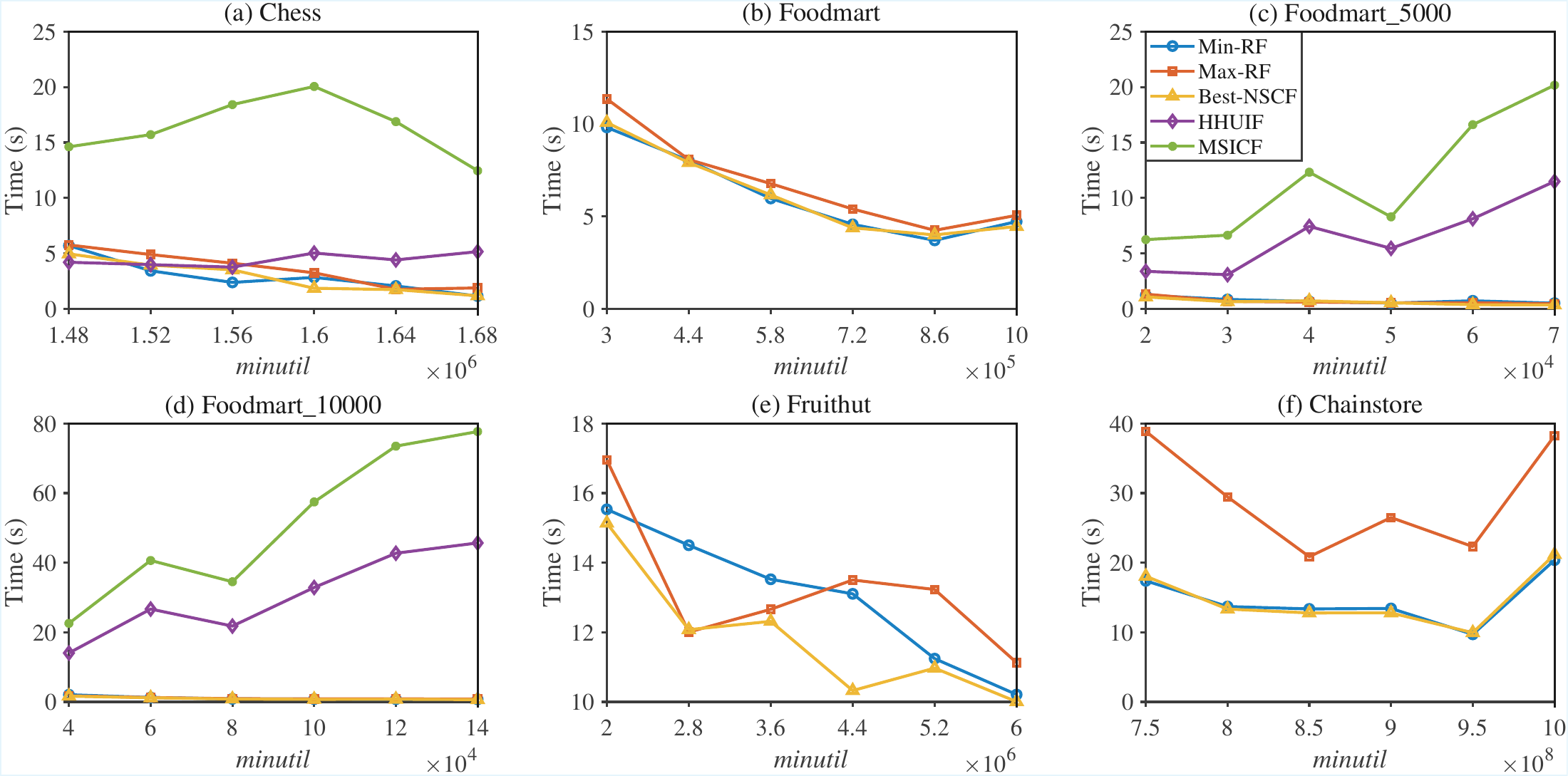} 
	\caption{Comparison of runtime under different minimum utility thresholds.}
	\label{runtime1} 
\end{figure*}
	
\begin{figure*}[htbp]
	\centering
	\includegraphics[width=0.8\textwidth]{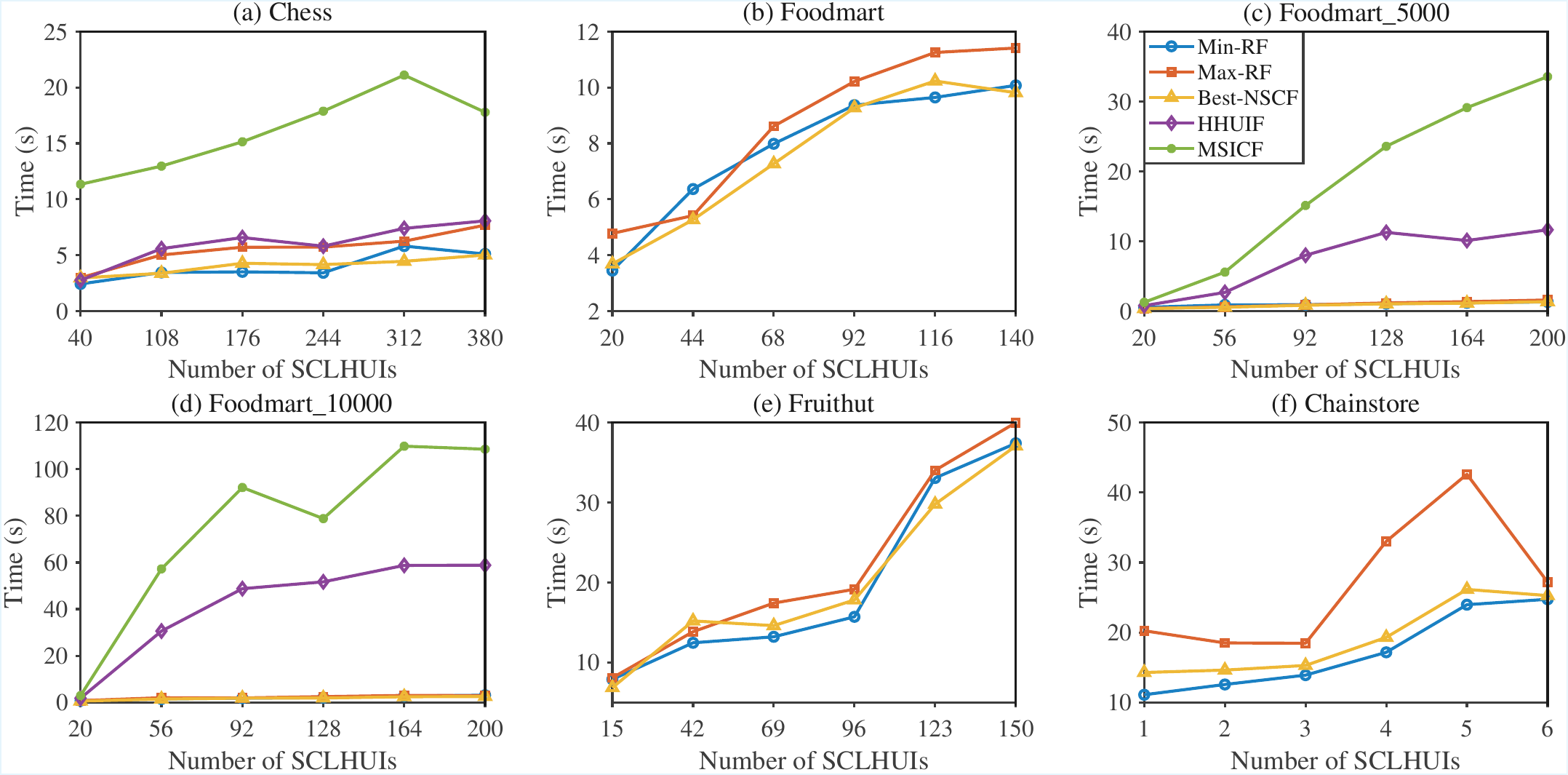} 
	\caption{Comparison of runtime under different numbers of sensitive itemsets.}
	\label{runtime2} 
\end{figure*}
	
\subsection{Runtime}
	
This section evaluates and compares the runtime of the three algorithms. Fig. \ref{runtime1} and Fig. \ref{runtime2} respectively present the runtime of different algorithms when varying \textit{minutil} and the number of sensitive itemsets.
	
Fig. \ref{runtime1} illustrates the runtime performance of each algorithm under different \textit{minutil} thresholds. As \textit{minutil} increases, the runtime of the algorithms decreases. This is because a higher \textit{minutil} results in a lower utility reduction required to hide SCLHUIs, thereby reducing the number of operations needed. On the five datasets other than Chainstore, the runtimes of Min-RF, Max-RF, and Best-NSCF are very similar. However, on the Chainstore dataset, the Max-RF algorithm, which is designed to improve efficiency, has a longer runtime than the other two algorithms, exceeding them by 56.2\% to 130.0\%. This is because, in CLPPUM, the algorithm transforms the processing of victim items into the processing of their leaf items. Although Max-RF reduces the number of operations on victim items by selecting those with the highest \textit{RGISU}, the victim items with the largest \textit{RGISU} often have a lower level and thus have a larger number of leaf items. This leads to a higher number of actual deletion operations compared to the other two algorithms. Additionally, the Chainstore dataset has the largest taxonomy, causing the runtime of Max-RF to increase significantly. For the smallest and densest dataset, Chess, the baseline method MSICF requires, on average 5 to 7 times more runtime than the proposed algorithms. This is mainly because MSICF needs to scan \textit{SC} to select victim items and continuously update the \textit{SC} of items during the hiding process. Moreover, the introduction of generalized items further enlarges the search space, which significantly increases the runtime. In contrast, Min-RF, Max-RF, and Best-NSCF benefit from the \textbf{GI-dic} dictionary, which reduces the computational complexity of related measures and thus improves efficiency. On the sparse datasets Foodmart\_5000 and Foodmart\_10000, the runtimes of HHUIF and MSICF are more than 10 times those of the proposed algorithms. Furthermore, the runtimes of these two baseline methods do not decrease as \textit{minutil} increases. This is because HHUIF and MSICF must reselect victim items after processing each one, and the presence of generalized items further expands the search space, making their runtime largely dependent on the number of victim item selections. Consequently, their runtimes not only increase substantially but also do not exhibit a consistent decreasing trend with increasing \textit{minutil}, instead varying according to the selected sensitive cross-level high-utility itemsets. For Foodmart, Fruithut, and Chainstore, HHUIF and MSICF fail to produce results within 10 minutes due to scalability limitations.

Fig. \ref{runtime2} illustrates the impact of varying the number of sensitive itemsets on the runtime of each algorithm. As the number of sensitive itemsets increases, the runtime of the algorithms increases. Increasing the number of sensitive itemsets requires the algorithms to hide more sensitive itemsets, thereby increasing their runtime. On the Chess dataset, the runtime of MSICF is on average 2–3 times that of Min-RF, Max-RF, and Best-NSCF. On the Foodmart and Fruithut datasets, the runtime of MSICF exceeds that of the three proposed algorithms by more than 30 times on average. On the Chainstore dataset, the runtime of Max-RF is 7.6\%–92.5\% higher than that of Min-RF and Best-NSCF. The differences among the five algorithms and their underlying reasons are largely consistent with those observed in Fig. \ref{runtime1}.

\begin{figure*}[htbp]
	\centering
	\includegraphics[width=0.8\textwidth]{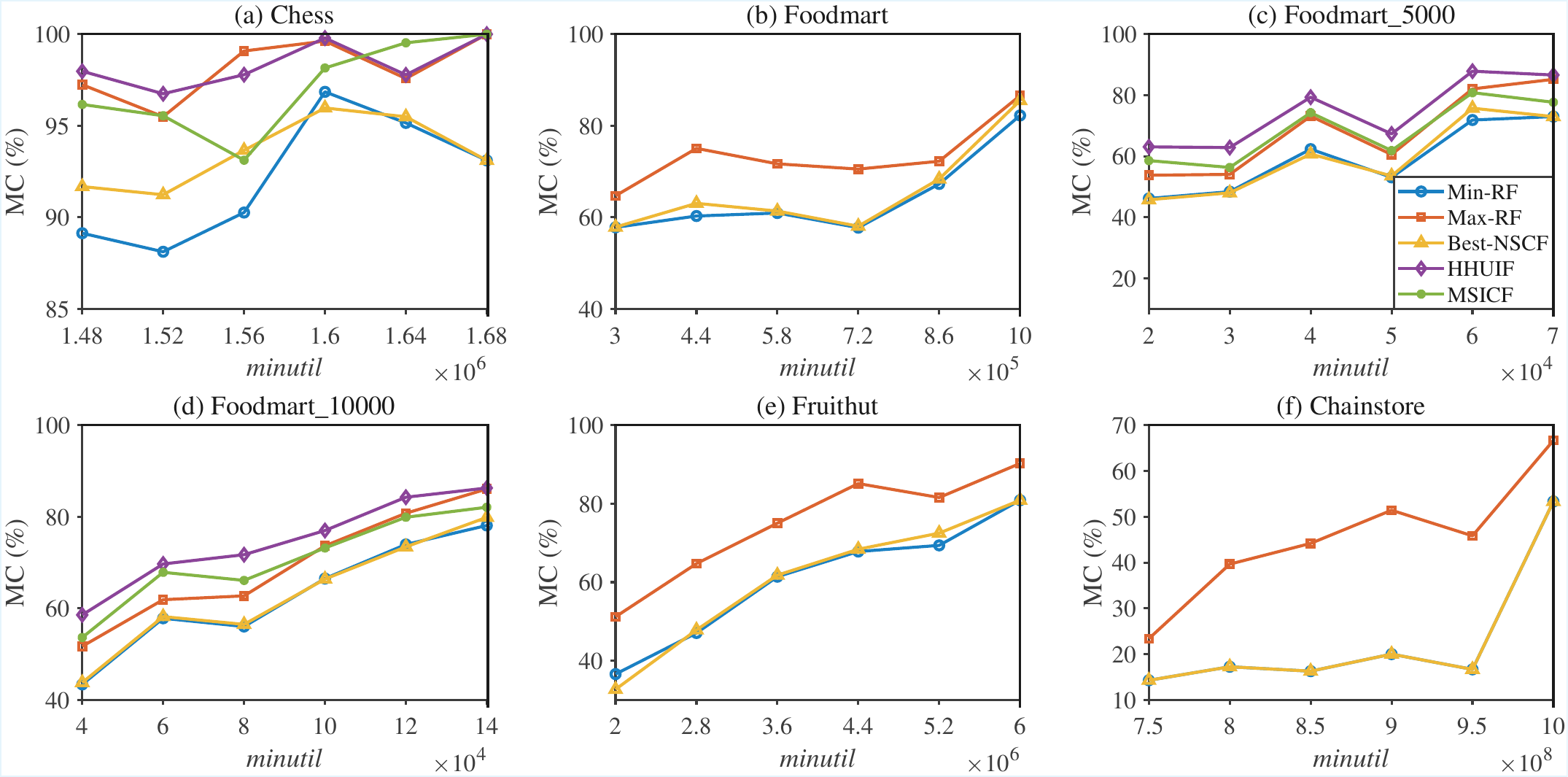} 
	\caption{Comparison of algorithm MC under different minimum utility thresholds.}
	\label{MC1} 
\end{figure*}
	
\begin{figure*}[htbp]
	\centering
	\includegraphics[width=0.8\textwidth]{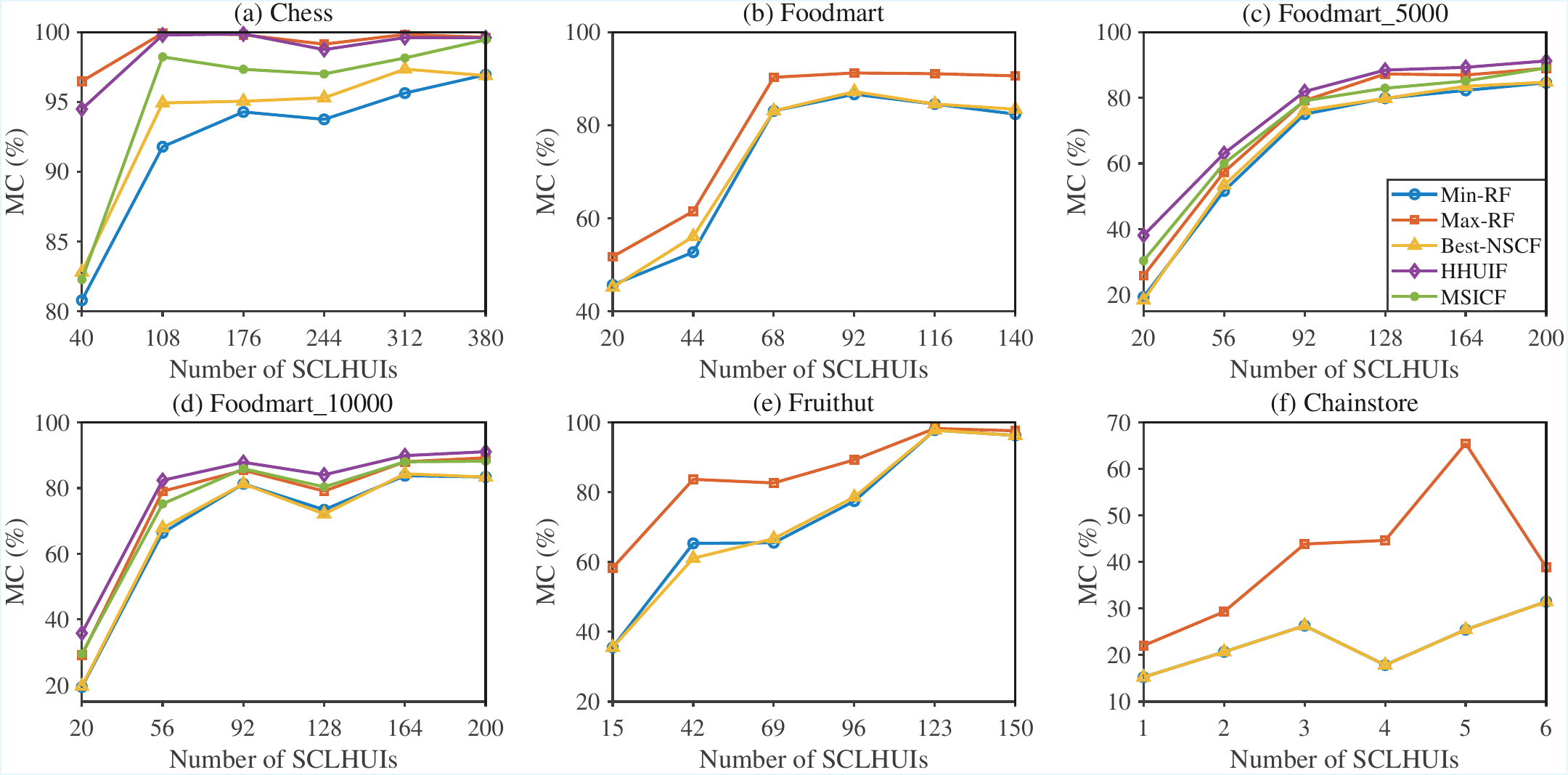} 
	\caption{Comparison of algorithm MC under different numbers of sensitive itemsets.}
	\label{MC2} 
\end{figure*}
		
\subsection{Missing cost}
	\begin{figure*}[htbp]
	\centering
	\includegraphics[width=0.8\textwidth]{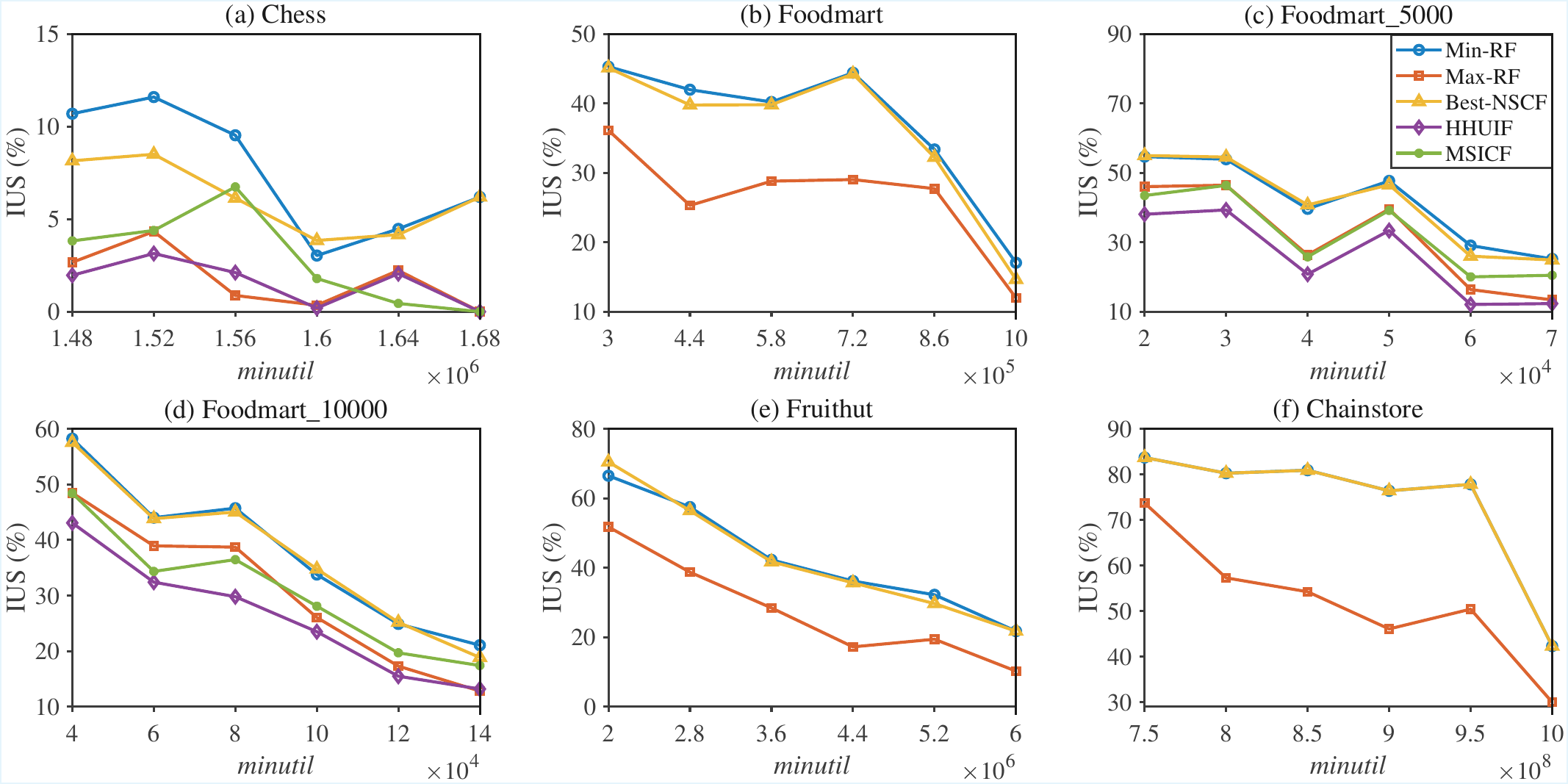} 
	\caption{Comparison of IUS under different minimum utility thresholds.}
	\label{IUS1} 
\end{figure*}

\begin{figure*}[htbp]
	\centering
	\includegraphics[width=0.8\textwidth]{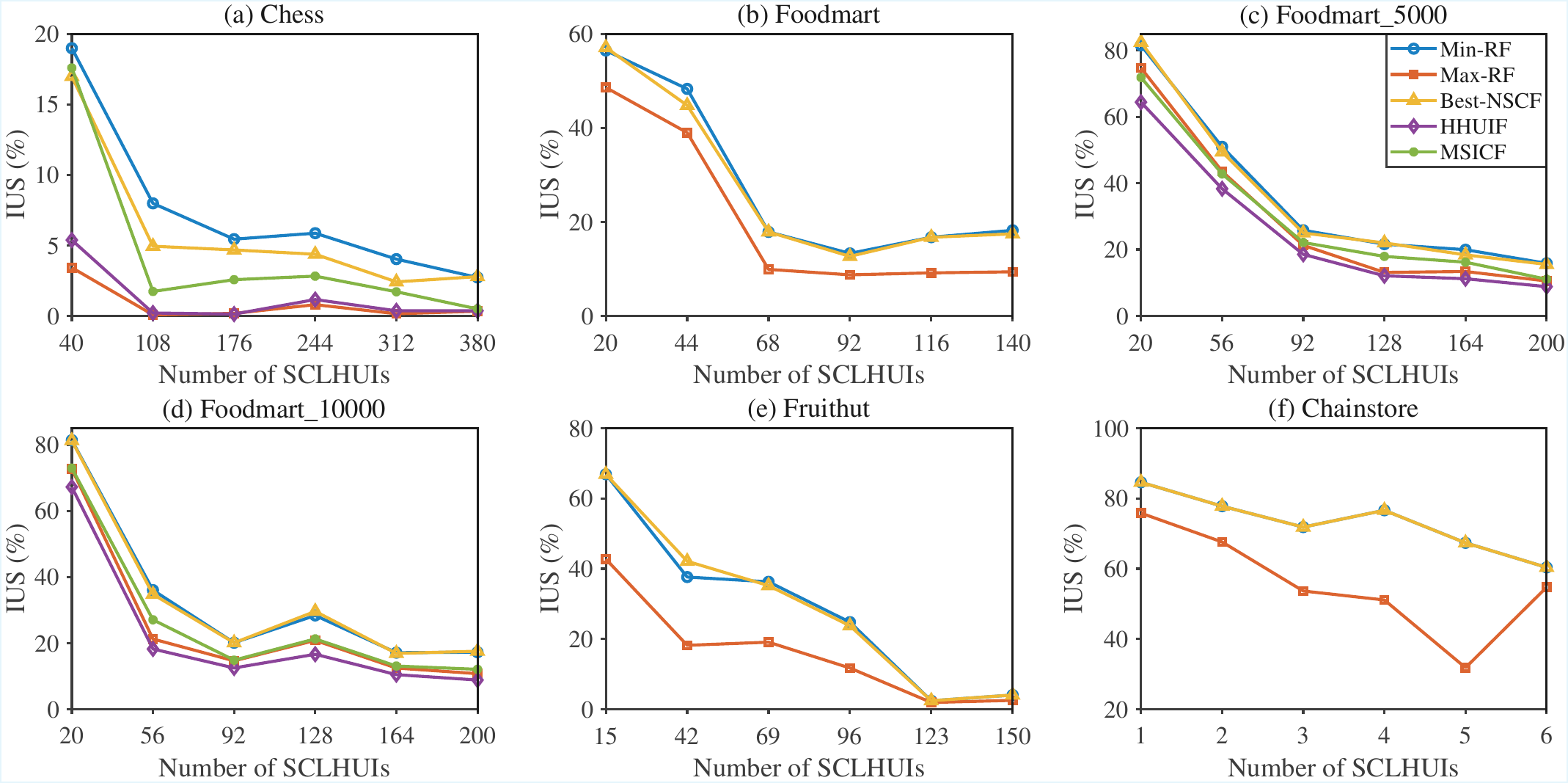} 
	\caption{Comparison of IUS under different numbers of sensitive itemsets.}
	\label{IUS2} 
\end{figure*}

In this section, we evaluate the MC of the Min-RF, Max-RF, and Best-NSCF algorithms, with the results shown in Fig. \ref{MC1} and Fig. \ref{MC2}. Due to the overlap between SCLHUIs and NSCLHUIs, algorithms based on item deletion cannot reduce the MC to 0.
	
Fig. \ref{MC1} shows the impact of changes in \textit{minutil} on the MC of each algorithm. As shown, the missing cost increases with the rise of the \textit{minutil}. This is because, as \textit{minutil} increases, fewer itemsets are discovered by CLHUIM, while the proportion of sensitive itemsets increases, leading to a continuous increase in MC. Across all datasets, the proposed Min-RF and Best-NSCF algorithms achieve the best performance in reducing MC compared with the other methods. Moreover, the MC of these two algorithms are almost identical on sparse datasets. This is mainly because most of the datasets are sparse, making it difficult for Best-NSCF to identify the optimal victim item simultaneously with respect to \textit{SC} and \textit{NSC}. As a result, Best-NSCF often behaves similarly to Min-RF under these conditions and exhibits comparable performance. On the Foodmart\_5000 dataset, the MC of Min-RF is on average 21.0\% and 13.8\% lower than that of HHUIF and MSICF, respectively. On Foodmart\_10000, the reductions are 20.6\% and 13.7\%. On the Foodmart and Fruithut datasets, the MC of Min-RF is on average 14.81\% and 25.7\% lower than that of Max-RF, respectively. This advantage comes from the strategy adopted by Min-RF, which selects the item with the smallest \textit{RGISU} as the victim item. Compared with utility alone, \textit{RGISU} better reflects the impact of modifying a victim item on other sensitive itemsets. On the Chess dataset,  when \textit{minutil} is reduced to the range of 1,480,000 to 1,560,000, Min-RF achieves a 0.6\% to 3.4\% lower MC compared to Best-NSCF. This is because Chess is a dense dataset, where Best-NSCF is better able to select victim items that are optimal in terms of \textit{NSC}. However, in CLHUIM, quantity (i.e., \textit{SC} and \textit{NSC}) provides less information than utility. Even if a (generalized) item has the fewest \textit{NSC} and the largest \textit{SC}, it may still affect a large number of NSCLHUIs due to its numerous leaf items and high utility. Conversely, \textit{RGISU} reflects not only the utility of a (generalized) item in sensitive transactions but also suggests, when lower, that the item resides at a lower level in the taxonomy and has fewer leaf items. Therefore, selecting the (generalized) items with the lowest \textit{RGISU} as the victim item can minimize the impact on database utility while accelerating the hiding process. This explains why Min-RF achieves both lower runtime and lower MC in the experiments.

Fig. \ref{MC2} shows the impact of changes in the number of sensitive itemsets on the MC of each algorithm. As the number of sensitive itemsets increases, the MC continues to rise. This is because the increase in SCLHUIs requires the algorithm to conceal more items, thereby leading to higher MC. On the five sparse datasets, the MC values of Min-RF and Best-NSCF are almost identical. On the Chess dataset, when the number of sensitive itemsets varies from 40 to 312, Min-RF achieves an MC that is 0.8\%–3.3\% lower than that of Best-NSCF. The differences among the five algorithms and the reasons for these differences are essentially the same as those shown in Fig. \ref{MC1}.
	
\begin{figure*}[htbp]
	\centering
	\includegraphics[width=0.8\textwidth]{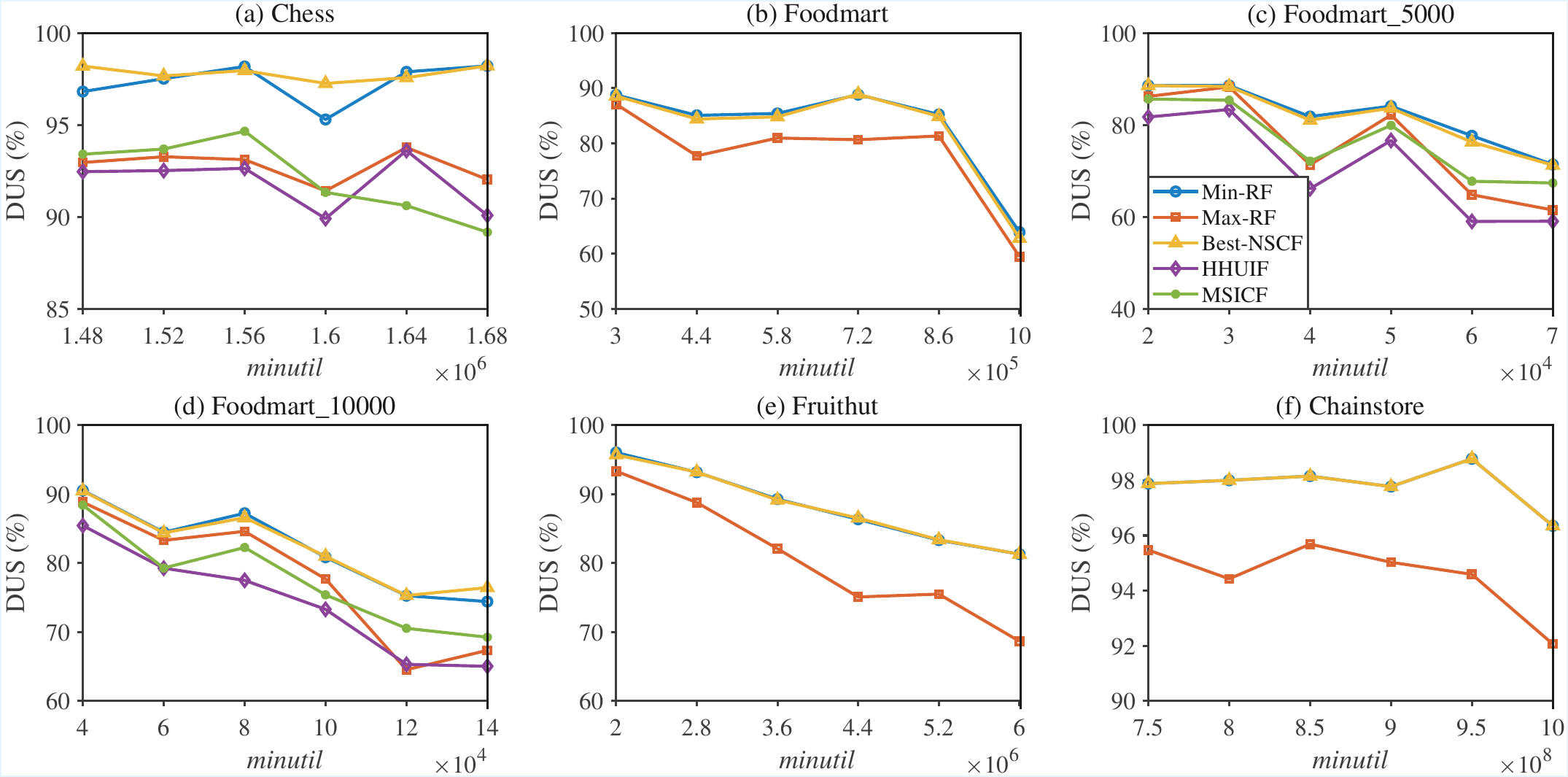} 
	\caption{Comparison of DUS under different minimum utility thresholds.}
	\label{DUS1} 
\end{figure*}
	
\begin{figure*}[htbp]
	\centering
	\includegraphics[width=0.8\textwidth]{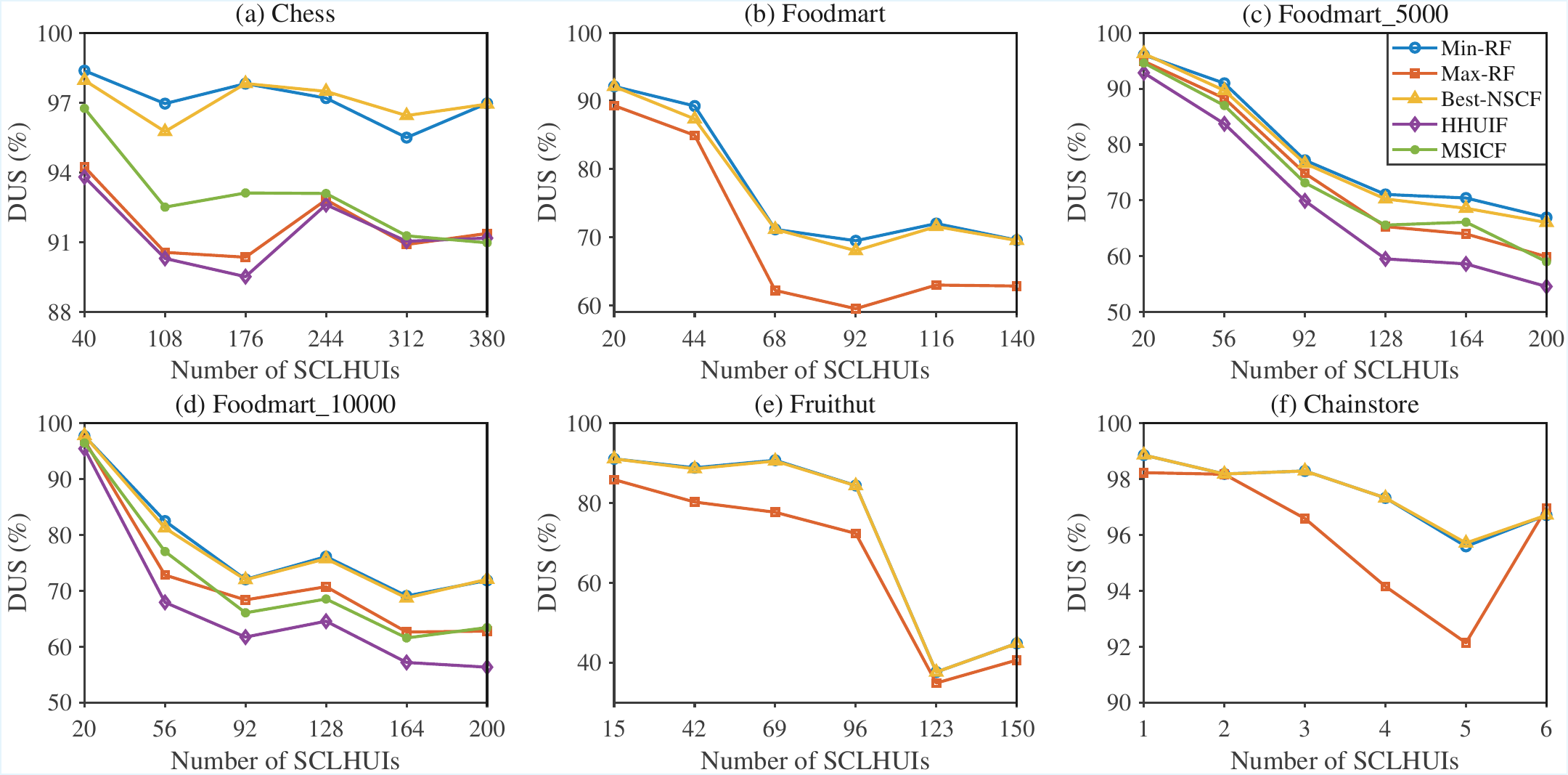} 
	\caption{Comparison of DUS under different numbers of sensitive itemsets.}
	\label{DUS2} 
\end{figure*}
	
\subsection{Itemset utility similarity}

This section presents an evaluation of IUS for the Min-RF, Max-RF, and Best-NSCF algorithms, with the results shown in Fig. \ref{IUS1} and Fig. \ref{IUS2}. Unlike MC, which focuses on the difference in the number of CLHUIs before and after hiding, IUS focuses on the change in the utility of CLHUIs. These two metrics exhibit an inverse relationship: as the MC increases, the IUS decreases accordingly.
	
Fig. \ref{IUS1} shows the impact of changing \textit{minutil} on the IUS of each algorithm. As \textit{minutil} increases, the IUS consistently decreases. This is because a higher \textit{minutil} results in fewer CLHUIs being discovered, increasing the proportion of sensitive itemsets among them. Consequently, the utility difference between the itemsets before and after sanitization becomes larger. Across all six datasets, Min-RF and Best-NSCF achieve the best IUS. On the Chess dataset, when \textit{minutil} falls within the range of 1,480,000–1,560,000, the IUS of Min-RF is 0.6\%–3.4\% higher than that of Best-NSCF. The differences in IUS among the five algorithms further support the conclusions drawn from MC: Min-RF and Best-NSCF perform best on sparse datasets, while on Chess, Min-RF achieves a slightly lower MC than Best-NSCF, leading to a higher IUS.
	
Fig. \ref{IUS2} shows the impact of changes in the number of sensitive itemsets on IUS of each algorithm. As the number of sensitive itemsets increases, the IUS continues to decrease. This is because the more sensitive itemsets the algorithm needs to hide, the fewer CLHUIs can be discovered from the sanitized database, resulting in lower IUS. On the sparse datasets, Min-RF and Best-NSCF achieve the best performance in terms of IUS. On the Chess dataset, the difference in IUS between Min-RF and Best-NSCF is similar to the difference observed in their MC values. The differences among the five algorithms and their underlying causes are largely consistent with those observed in Fig. \ref{IUS1}.

\begin{figure*}[htbp]
	\centering
	\includegraphics[width=0.8\textwidth]{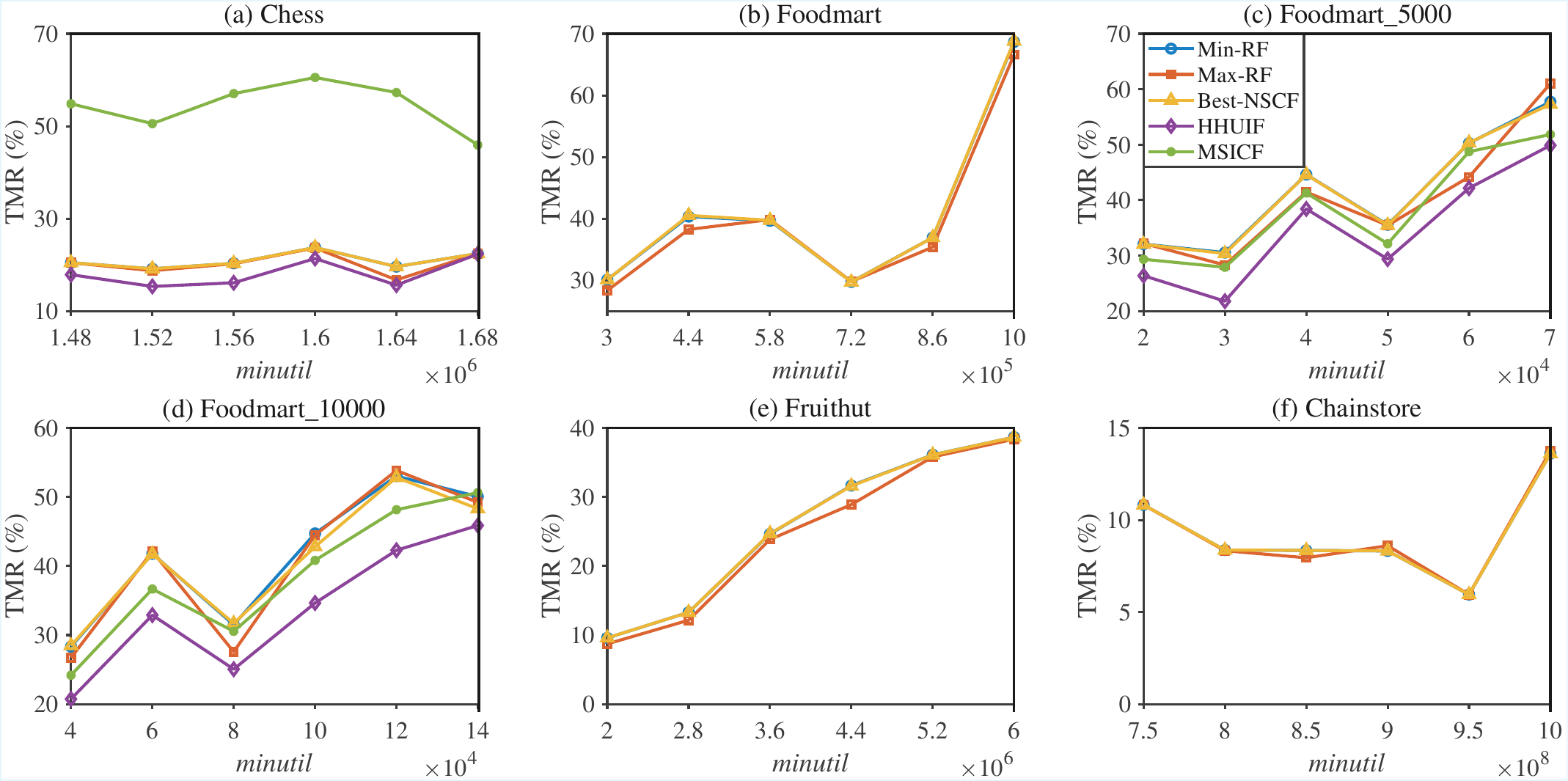} 
	\caption{Comparison of TMR under different minimum utility thresholds.}
	\label{TMR1} 
\end{figure*}
	
\begin{figure*}[htbp]
	\centering
	\includegraphics[width=0.8\textwidth]{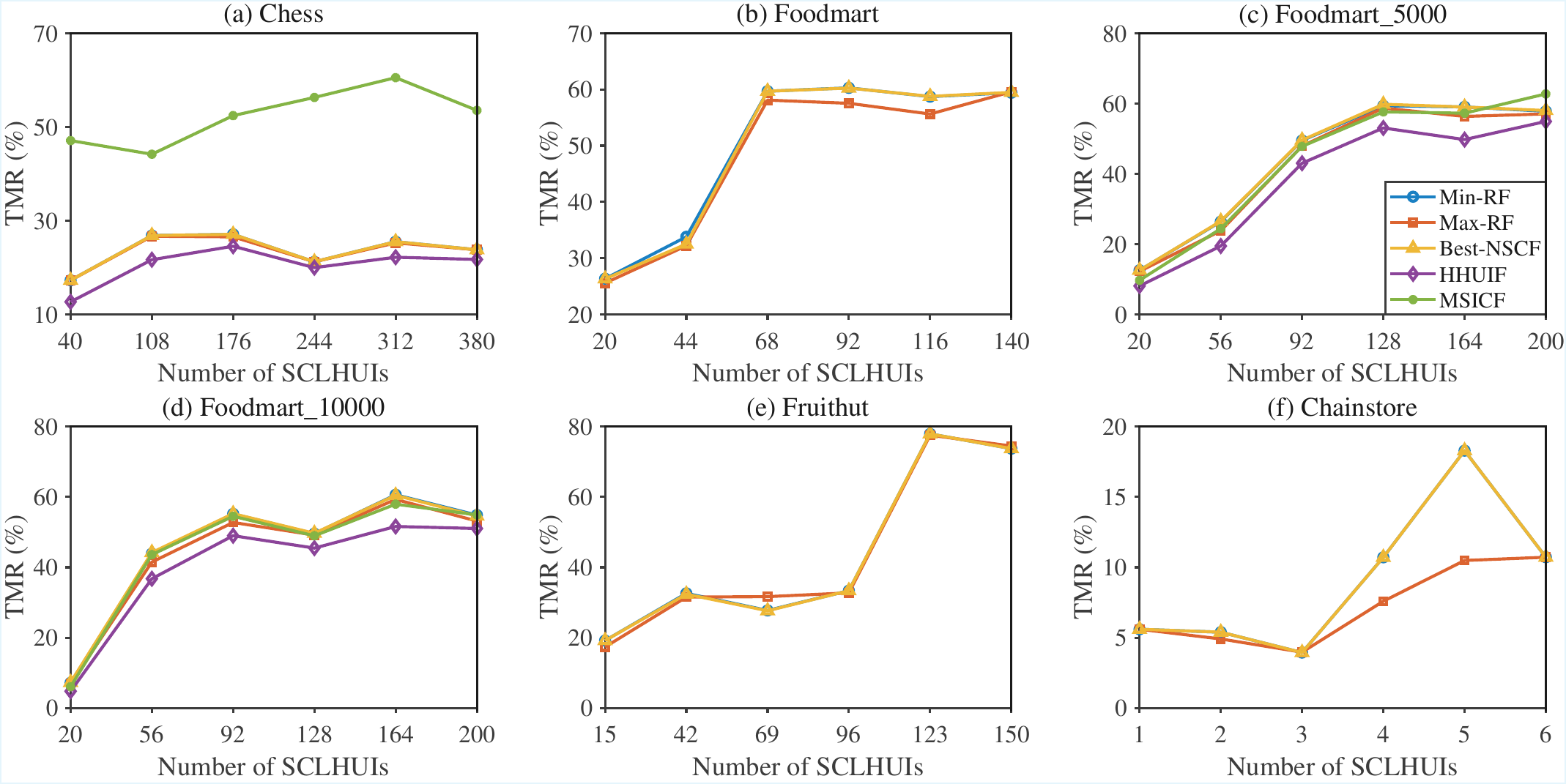} 
	\caption{Comparison of TMR under different numbers of sensitive itemsets.}
	\label{TMR2} 
\end{figure*}
	
\subsection{Database utility similarity}
	
In this section, we evaluated the DUS of the Min-RF, Max-RF, and Best-NSCF algorithms. The results are shown in Fig. \ref{DUS1} and Fig. \ref{DUS2}.  The results from the two experimental approaches show a high degree of consistency. As \textit{minutil} increases, the utility of the selected SCLHUIs also increases. These itemsets often contain more generalized items with lower levels, which leads to an increase in deletion operations and consequently a decline in DUS. As the number of SCLHUIs increases, the algorithm also needs to delete more itemsets, resulting in a continuous decrease in DUS. Across the six datasets, Min-RF and Best-NSCF achieve the best performance. On the Chess dataset, when \textit{minutil} is 1,480,000, and the number of sensitive itemsets is 312, although Min-RF achieves better MC and IUS than Best-NSCF, the DUS of Best-NSCF is 1.4\% and 0.9\% higher than that of Min-RF, respectively. This is due to the randomness in selecting SCLHUIs, which means that although Min-RF only seeks the (generalized) items with the smallest \textit{RGISU} as victim items, in certain cases, this strategy tends to favor generalized items with lower \textit{SC}, which may require the deletion of more victim items to successfully hide the sensitive itemsets. On dense datasets, Best-NSCF can relatively more easily identify items with optimal \textit{SC} and \textit{NSC}. Best-NSCF selects the (generalized) items with the smallest \textit{NSC} values, minimizing their impact on NSCLHUIs, resulting in a higher DUS compared to the Min-RF algorithm.

\subsection{Transaction modification ratio}
	
Finally, we evaluated the TMR of the Min-RF, Max-RF, and Best-NSCF algorithms, with the results shown in Fig. \ref{TMR1} and Fig. \ref{TMR2}. It can be observed that, on all datasets except Chainstore, the TMR of Min-RF, Max-RF, and Best-NSCF remain relatively similar and relatively low. This is because these algorithms process transactions in a fixed order, which reduces scanning overhead and tends to concentrate modifications on transactions that have already been modified. On the Chess dataset when \textit{minutil} = 1,640,000, and on the Chainstore dataset when the number of sensitive itemsets is 4 or 5, Max-RF achieves a lower TMR than the other two algorithms. This is because Max-RF selects the (generalized) item with the largest RGISU as the victim item. Although Max-RF performs more deletion operations on leaf items than the other two algorithms, it requires fewer operations on the victim items themselves and affects fewer transactions containing those items, resulting in a lower TMR. Among the five algorithms, HHUIF produces the lowest TMR. This is because it selects transactions for modification based on utility, causing the modified transactions to be concentrated among those with higher \textit{TU}. In contrast, MSICF exhibits a relatively high TMR on the dense Chess dataset. On dense datasets, strong correlations exist among sensitive itemsets, leading to very similar \textit{SC} values across items. As a result, selecting victim items solely based on \textit{SC} provides limited guidance. Moreover, since MSICF does not process transactions in a fixed order, it often modifies previously untouched transactions while repeatedly handling victim items, which further increases the TMR.

\section{Conclusion} \label{sec:conclusion}

Existing PPUM methods often overlook the hierarchical information inherent in real-world data. To address this gap, this paper formulates the novel task of CLPPUM and designs three new algorithms called Min-RF, Max-RF, and Best-NSCF to hide SCLHUIs. To effectively evaluate the impact of sanitization operations involving generalized items, key metrics (\textit{SC}, \textit{NSC}, and \textit{RISU}) are extended into the cross-level context. Additionally, a dedicated \textit{GI-dic} structure is designed to accelerate metric computation and assist in victim item selection, thereby improving the overall efficiency of the algorithms. Experiments on multiple datasets show that all three proposed algorithms can successfully hide all SCLHUIs without generating false itemsets. Among them, Min-RF demonstrates the best overall performance. However, the proposed method has certain limitations on dense datasets: due to strong correlations among items, the deletion-based strategy affects a large number of non-sensitive itemsets, leading to higher side effects (measured by MC). Therefore, achieving a better balance among mining efficiency, privacy preservation, and side-effect minimization remains an unresolved challenge. Future work will focus on developing more flexible victim item selection and database modification strategies to further reduce side effects and improve performance on dense datasets. Promising directions include integrating multi-objective optimization frameworks to explicitly trade off different metrics, as well as exploring parallel or distributed computing paradigms to enhance scalability.

\section*{Acknowledgment}
This research was supported in part by the National Natural Science Foundation of China (No. 62272196), and Guangzhou Basic and Applied Basic Research Foundation (No. 2024A04J9971).

\section*{Data Availability}
Datasets and code are available at \url{https://github.com/jhcai321/CLPPUM}.

\section*{CRediT Authorship Contribution Statement}
Jiahong Cai: Methodology, Writing original draft.
Wensheng Gan: Review and editing, Supervision.
Philip S. Yu: Review and editing.


\bibliographystyle{cas-model2-names}
	
\bibliography{CLPPUM.bib}

\end{document}